\documentclass[conference]{IEEEtran}

\usepackage[T1]{fontenc}
\usepackage{hhline}
\usepackage{amssymb,amsfonts,amsthm}
\usepackage{amsmath}
\usepackage{epsfig}
\usepackage{color}
\usepackage{enumerate}
\usepackage{ifpdf}
\usepackage{graphicx}
\usepackage{subfig}
\newtheorem{theorem}{Theorem}

\theoremstyle{remark}

\theoremstyle{definition}

\newtheorem{proposition}{Proposition}

\begin{document}

\title{Design of Virtualized Network Coding Functionality for Reliability Control of Communication Services over Satellite}

\author{Tan Do-Duy, M. A. V\'azquez-Castro\\
Dept. of Telecommunications and Systems Engineering\\
Autonomous University of Barcelona, Spain\\
Email: \{tan.doduy, angeles.vazquez\}@uab.es}

\maketitle
\begin{abstract}
Network coding (NC) is a novel coding technology that can be seen as a generalization of classic point-to-point coding. As with classic coding, both information theoretical and algebraic views bring different and complementary insights of NC benefits and corresponding tradeoffs. However, the multi-user nature of NC and its inherent applicability across all layers of the protocol stack, call for novel design approaches towards efficient practical implementation of this technology.

In this paper, we present a possible way forward to the design of NC as a virtual network functionality offered to the communication service designer. Specifically, we propose the integration of NC and Network Function Virtualization (NFV) architectural designs. The integration is realized as a toolbox of NC design domains that the service designer can use for flow engineering. Our proposed design framework combines network protocol-driven design and system modular-driven design approaches. 
In particular, the adaptive choice of the network codes and its use for a specific service can then be tailored and optimized depending on the ultimate service intent and underlying (virtualized) system or network.

We work out a complete use case where we design geo-network coding, an application of NC for which coding rate is optimized using databases of geo-location information towards an energy-efficient use of resources. Our numerical results highlight the benefits of both the proposed NC design framework and the specific application.

\end{abstract}

\begin{IEEEkeywords}
Network coding, network function virtualization, satellite communications, reliability, connectivity.
\end{IEEEkeywords}

\IEEEpeerreviewmaketitle

%------------------------------------------------------------------------------------%
\section{Introduction}
\vspace{-2pt}

Fundamental results show that NC technology enables network capacity achievability both for noiseless \cite{Ahlswede.2000} and noisy networks \cite{Dana.2006}. Such fundamental ground has lead to an enormous body of research which further confirms the great potential of the technology as a multipurpose communication networking tool to control capacity and reliability across different types or networks and systems with applicability at different communication layers. 

Full exploitation of NC potential however clearly requires novel design paradigms as compared to traditional networking protocols. This is due to the fact that differently to traditional networking procedures, NC optimizes the information flow by considering such flow as a non-physical flow which is subject of computational manipulation. NC should also operate in full coordination with other existing and upcoming networking protocols. It becomes self-evident that such new design paradigms should not rely on dedicated and/or specialized physical hardware or network-dependent devices. Such requirements naturally match current trends of virtualization of standard technologies of computation networking and storage. 

As an innovative technique towards the implementation of network functions, network function virtualization (NFV) is clearly a NC design option as it would make NC available as a flow engineering functionality offered to the network. NFV has been proposed as a promising design paradigm by the telecommunications sector to facilitate the design, deployment, and management of networking services (see \cite{Han.2015, Mijumbi.2016, Liang.2015} for details). Essentially, NFV separates software from physical hardware so that a network service can be implemented as a set of virtualized network functions through virtualization techniques and run on commodity standard hardware. Furthermore, NFV can be instantiated on demand at different network locations without requiring the installation of new equipment.
The integration of NC and NFV will enable the applicability of NC in future networks (e.g. upcoming 5G networks \cite{Fallgren.2013}) to both distributed (i.e. each network device) and centralized manners (i.e. servers or service providers). The European Telecommunications Standards Institute (ETSI) has proposed some use cases for NFV in \cite{NFV001.2013}, including the virtualization of cellular base station, mobile core network, home network, and fixed access network, etc. Furthermore, NFV is expected as an innovative technology to promote cost-efficient and new added-value services for satellites \cite{Gardikis.2016}.

There are already available proposals that combine NC virtualization and software-defined networking (SDN) technology. For example, in \cite{Szabo2.2015}, the authors investigate the usability of random linear NC as a function in SDNs to be deployed with virtual (software) OpenFlow switches. The NC functions including encoder, re-encoder, and decoder are run on a virtual machine (ClickOS). This work provides a prototype with implementation of additional network functions via virtual machines (VMs) by sharing system resources or additional hardware (e.g. FPGA cards). In \cite{Hansen.2015}, NC is implemented in a VM which is then embedded into an Open vSwitch. The paper shows the relation between NC (as a software), VM and host OS of Open vSwich. These works indicate the feasibility of integrating NC as a functionality based on SDN and the concentration of network functions in centralized architectures such as data centers or centralized locations proposed by network operators, service integrators and providers. However, a unified design framework for NC design in view of NFV either centralized or distributed is currently missing. In this paper, we focus on the functional level regardless of how it is softwarised.

The motivation of this paper is to show the generalization of NC design domains as a toolbox so that NC can be applied over different operational frameworks and services including satellite or hybrid networks thus enabling softwarization and rapid innovative deployment. Besides, the use of databases with geo-tagged link statistics and network nodes' location may provide useful information for the optimization of network performance.
Our contributions can be summarized as follows.
\begin{itemize}
	\item We propose the integration of NC and NFV as a functional architecture given as a structured toolbox of NC design domains so that tailored designs can be implemented in software and deployed over virtualised infrastructures to allow flow engineering tailored to service operational intent.
	\item We describe the coding design domain with coherent and non-coherent design approaches. In particular, we consider systematic NC as an example of coherent approach. On the other hand, a novel subspace coding scheme, as part of non-coherent approach, is introduced to reveal the analytical convenience of $\mathbb{F}_q$-linear erasure channels while enabling the use of linear algebra for efficient encoding, re-encoding, and decoding.
	\item We validate our proposed NC function (NCF) over a complete design for the use case of geo-network coding. VNCF is integrated as a VNF in ETSI NFV architecture with clear identification of exchanges between VNCF and NFV-Management and Orchestration via reference points. A general procedure for the instantiation, performance monitoring, and termination of VNCF is provided. In particular, coding is optimized according to geo-tagged link statistics and geo-location information given computational complexity/energy efficiency constraints.
	Specifically, numerical results indicate that per-link achievable rate region obtained with NC at the source and re-encoding at the intermediate node is twice wider than that of transmission with NC at the source only. Furthermore, connectivity gains up to $32$ times and $47$ times are obtained if compared with the uncoded case for $80\%$ and $85\%$ target reliability with computational complexity constraint of $125000$ operations in terms of multiplications and additions, respectively.
\end{itemize}

The rest of this paper is organized as follows.
In Section \ref{sec:NCframework}, we propose the architectural design of NC. 
Coding design domain with coherent and non-coherent design approaches is introduced in Section \ref{sec:Fqlinear}.
In Section \ref{sec:NCFV}, we present our proposed integrated design of NC and existing NFV architecture.
In Section \ref{sec:CaseStudy}, we validate our proposed architectural design in a specific virtualized Geo-NC function (VGNCF) design using the functionalities identified in previous sections.
In Section \ref{sec:Numerical}, we conduct numerical results to validate the performance for the case of using our VGNCF design and the uncoded case.
Finally, Section \ref{sec:Conc} identifies conclusions and further work.
\vspace{-6pt}

%------------------------------------------------------------------------------------%
\section{NC DESIGN FRAMEWORK}
\label{sec:NCframework}
\vspace{-2pt}
Virtualization and NC are two different techniques to address different challenges in the designs of upcoming network technologies. The combination of virtualization and NC brings forth a potential solution for the management and operation of the future networks \cite{Hansen.2015}. NC design involves different domains \cite{MAVazquez2.2015}:

\begin{itemize}
	\item \textbf{NC coding domain} - domain for the design of network codebooks, encoding/decoding algorithms, performance benchmarks, appropriate mathematical-to-logic maps, etc. Note that this is a domain fundamentally related to coding theory.
	\item \textbf{NC functional domain} - domain for the design of the functional properties of NC to match design requirements built upon abstractions of 
\begin{itemize}
	\item \textbf{Network}: by choosing a reference layer in the standardized protocol stacks and logical nodes for NC and re-encoding operations. 
	\item \textbf{System}: by abstracting the underlying physical or functional system at the selected layer e.g. SDN and/or function virtualization. 
\end{itemize}

Note that this is a domain fundamentally related to electrical engineering and computer science. However, note that while traditionally electrical engineering has been concerned with physical networks and computer science with logical networks, virtualization brings both disciplines closer to each other.
\item \textbf{NC protocol domain} - domain for the design of physical signaling/transporting of the information flow across the virtualized physical networks in one way or interactive protocols. 
\end{itemize}

The domain clearly relevant for NC to be designed as a NFV is the NC functional domain and we develop this domain in the next section. 

\vspace{-6pt}
%------------------------------------------------------------------------------------%
\section{CODING DESIGN DOMAIN} 
\label{sec:Fqlinear}
\vspace{-2pt}
Consider that a source node has $k$ data packets to send towards a sink node. The source is connected to the sink via intermediate node(s). Each packet is a column vector of length $m$ over a finite field $\mathbb{F}_q$. The set of data packets in matrix notation is $S=[s_1 \hspace{0.1cm} s_2 \hspace{0.1cm} ... s_k]$, where $s_t$ is the $t^{th}$ data packet. The links are modeled as memoryless erasure channels with erasure probability $\epsilon_{i}$, $i=1,2,..$.

\subsection{Coherent Network Coding} 
\label{sec:SNC}
As part of coherent NC, systematic NC (SNC) has been considered as a practical NC scheme to increase network reliability of wireless networks. SNC is derived in \cite{Saxena.2016} and introduced here for convenience.

\subsubsection{Encoding at the source node:}
The SNC encoder sends $k$ data packets in the first $k$ time slots (systematic phase) followed by $n-k$ random linear combinations of data packets in the next $n-k$ time slots (non-systematic phase). Let $X=SG$ represents $k$ systematic packets and $n-k$ coded packets transmitted by the SNC encoder during $n$ consecutive time slots. The generator matrix $G = [I_k \hspace{0.1cm} C]$ consists of the identity matrix $I_k$ of dimension $k$ and $C \in \mathbb{F}^{k\times (n-k)}_q$ with elements chosen randomly from a finite field $\mathbb{F}_q$. The code rate is given by $r=k/n$.

\subsubsection{Re-encoding at the intermediate node:}
The intermediate node stores all the received packets in its buffer. The SNC re-encoder performs re-encoding operations in every time slot and sends n packets to the sink node. Let $X_I = XD_1T$ represents $n$ packets transmitted by the SNC re-encoder during $n$ consecutive time slots where $D_1 \in \mathbb{F}_q^{n\times n}$ represents erasures from the source node to the intermediate node and $T \in \mathbb{F}_q^{n\times n}$ represents the re-encoding operations at the intermediate node.

The erasure matrix $D_1$ is a $n\times n$ diagonal matrix with every diagonal component zero with probability $\epsilon_{1}$ and one with probability $1-\epsilon_{1}$. The re-encoding matrix $T$ is modeled as follows. During the systematic phase, if a packet $s_t$ is lost i.e., $D_1(t, t)=0$ then the non-zero elements of the $t_{th}$ column of matrix $T$ are randomly selected from $\mathbb{F}_q$. This represents that if the systematic packet is lost from the source node to the intermediate node, then the intermediate node transmits a random linear combination of the packets stored in its buffer. If a packet $s_t$ is not lost; i.e., $D_1(t, t)=1$ then the $t_{th}$ column of matrix $T$ is the same as the $t_{th}$ column of identity matrix $I_n$. This represents that the intermediate node forwards this systematic packet to the sink. During the non-systematic phase, the intermediate node sends a random linear combination of the packets stored in its buffer and all the non-zero elements of last $n-k$ columns of $T$ are chosen randomly from the finite field $\mathbb{F}_q$.

\subsubsection{Decoding at the sink node:}
Let $Y = X_ID_2$ represent $n$ packets received by the sink node where $D_2$ represents erasures from the intermediate node to the sink node. $D_2$ is $n\times n$ diagonal matrix of the same type as $D_1$ but with erasure probability $\epsilon_{2}$. If the sink does not receive any packet in time slot $t$ then the $t_{th}$ column of $Y$ is a zero column. 

The overall SNC coding strategy can be expressed using a linear operation channel (LOC) model where the output at the sink is $Y = SGH$ where $H = D_1TD_2$ represents the transfer matrix of the network. We assume that the coding vectors are attached in the packet headers so that the matrix $GH$ is known at the sink. The decoding is progressive using Gauss-Jordan elimination algorithm. All the $k$ data packets are recovered when $k$ innovative packets are received at the sink, i.e., $rank(GH)=k$.

Assuming Bernoulli distributed erasures and using the probability mass function $\alpha\left(j,v,p\right)={\binom{v}{j}} \left(1-p\right)^{j}\left(p\right)^{v-j}$, the residual erasure rate after decoding at each single hop is derived in \cite{Saxena.2016} and introduced here for convenience:
\begin{eqnarray}
\eta_{i}\left(r,\epsilon_{i}\right)=\epsilon_{i}\left(\phi_{1}\left(r,\epsilon_{i}\right)+\phi_{2}\left(r,\epsilon_{i}\right)\right), 
\label{eq:E7}
\end{eqnarray}
with 

$\phi_{1}\left(r,\epsilon_{i}\right)=\sum_{l=0}^{k-1}\alpha\left(l,n-1,\epsilon_{i}\right)$,

$\phi_{2}\left(r,\epsilon_{i}\right)=\sum_{l_{1}=0}^{k-1}\alpha\left(l_{1},k,\epsilon_{i}\right) \sum_{l_{2}=k-l_{1}}^{n-k}\alpha\left(l_{2},n-k,\epsilon_{i}\right)\times$
$\left(1-\text{\ensuremath{\prod}}_{l_{3}=0}^{k-l_{1}-1}\left(1-q^{l_{3}-l_{2}}\right)\right)$.

\subsection{Incoherent Network Coding}
In \cite{Koetter.2008, Silva.2010}, subspace coding is proposed for improving the reliability of $\mathbb{F}_q$-linear channels. The transmission is assumed incoherent and so the underlying (linear) NC problem is assumed to be separately solved. Then, subspace coding provides error correction over such q-linear network coded transmission channel. Subspace codes proposed in \cite{Koetter.2008, Silva.2010} follow (classic) distance metric driven constructions and can correct errors and erasures. Specifically, proposed subspace codes are constructed as lifted rank metric codes. In the following, lifted Gabidulin codes are described followed by the subspace coding proposed in this work.

\subsubsection{Lifted rank-metric codes:}
Rank-metric codes are matrix codes. Let $F^{n\times m}_q$ be the set of all $n\times m$ matrices with entries in $\mathbb{F}_q$. The rank distance between two matrices A and B is defined as $d_r(A,B)=rank(A-B)$. A rank-metric code $\mathcal{C}$ is defined as a subset of $F^{n\times m}_q$. This definition considers that matrices are codewords and these codes correspond to the so-called Delsarte codes. Instead, if codewords are considered as vectors with entries in the extension field $F_{q^m}$, then the corresponding rank-metric codes are the so-called Gabidulin codes, which are then a special case of Delsarte codes. The minimum rank distance of a rank-metric code $\mathcal{C}$ is denoted as $d_R(\mathcal{C}$). The Singleton bound for rank metric codes is
\begin{align*}
\left|\mathcal{C} \right|\leq q^{max\left\{n,m\right\}}(min\left\{n,m\right\}-d+1)
\end{align*}

for every code $\mathcal{C} \subset F^{n\times m}_q$ such that $d_R(\mathcal{C}=d$. Codes that achieve this bound are called maximum rank-distance (MRD) codes and linear MRD codes are known to exist for all choices of parameters $q, n, m$ and $d\leq min(n,m)$. It can be shown that Gabidulin codes can be seen as the rank-metric equivalent ($q^m$-ary analog) of the Reed Solomon codes and are MRD codes.

The construction of subspace codes given in \cite{Silva.2010} for constant dimension subspace code, $\mathcal{C}_S$, consists of lifting $\mathcal{C}_\mathcal{G}$ as follows
\begin{align*}
\mathcal{C}_\mathcal{S}=Lift(\mathcal{C}_\mathcal{G})\triangleq \left\{ \left\langle \left[ I_{k\times k} \hspace{0.3cm} X\right]\right\rangle \in F^{k\times k+m}_q, X \in \mathcal{C}_\mathcal{G} \right\}.
\end{align*}

Clearly, $\left|\mathcal{C}_\mathcal{G}\right|=\left|\mathcal{C}_\mathcal{S}\right|$ and it can be shown that lifted MRD codes are asymptotically optimal constant-dimension codes. The disadvantage of lifted rank-metric codes is however that information data must be mapped to the codewords and a basis must be chosen appropriately. Moreover, decoding algorithms are based on the rank metric between codewords and use of linear algebra is convolved. In the following section a more convenient subspace coding scheme is presented with fast encoding and efficient on-the-fly decoding algorithms using simple linear algebra.

\subsubsection{Subspace Coding for $\mathbb{F}_q$-Linear Erasure Satellite Channels:} 
In this work subspace coding is proposed based on random linear coding (RLC).

\begin{theorem}
Subspace encoding the $\mathbb{F}_q$-linear multicast erasure satellite channel. Let $X\in F^{k\times m}_q$ a set of $k$ information packets. A capacity-achieving subspace coding scheme for the $\mathbb{F}_q$-linear multicast satellite channel is given by the linear map
\begin{align*}
F^{k\times m}_q  				\longrightarrow F^{n\times m}_q: \hspace{0.3cm}		X			\longmapsto G_sX
\end{align*}

where $G_s\in F^{n\times k}_q$ is the channel encoding matrix applied at the source. The channel encoding matrix is a lifted random matrix, $G_s=Lift(H)$ where the elements of $H$ are randomly chosen in $\mathbb{F}_q$ and $H\in F^{(n-k)\times m}_q$.
\end{theorem}

\begin{proof}
The matrix $H$ corresponds to the capacity achieving random linear coding (RLC) scheme, which is proved capacity-achieving in \cite{Dana.2006} for wireless erasure networks. In \textit{Proposition} 1 it is shown that it is also a channel coding scheme when routing only achieves capacity. The lifting just introduces an additional algebraic property. Namely, any $\mathbb{F}_q$-linear transformation of the encoded data remains row equivalent to the injected information data packets, hence allowing non-coherent transmission as in \cite{Koetter.2008}. 
\end{proof}

The usefulness of the $\mathbb{F}_q$-linear model becomes now clear as note that the transmission with NC towards every multicast destination can be described by the following composition of linear maps
\begin{align*}
F^{k\times m}_q  		\longrightarrow F^{n\times m}_q 					\longrightarrow	F^{l\times m}_q: \hspace{0.3cm}		X		\longmapsto G_sX				\longmapsto	H_{r_i}G_sX
\end{align*}

with $H_{r_i}\in F^{l\times n}_q$ is the $\mathbb{F}_q$-linear channel transfer to every destination. 

\begin{theorem}
Subspace decoding for the $\mathbb{F}_q$-linear multicast erasure satellite channel. Assuming that $k$ independent coded packets are retrieved at every destination, the decoding consists of progressive Gauss-Jordan elimination. 
\end{theorem}

\begin{proof}
The transmitted packets are linearly transformed by the stabilizers while the channel erases full packets at every link with probability $\epsilon_i$. Hence, the overall linear transformation over the encoded packets retrieved at the sink node can be represented by a succession of invertible matrices, i.e. a sequence of elementary row operations. Assuming that $k$ independent encoded packets are retrieved, by basic linear algebra such sequence of matrices is row equivalent to the injected coded packets $G_sX$. This is because the first $k\times k$ minor is non-singular and matrices are classified up to equivalence by their rank. Hence, the set can be progressively decoded at the receiver, as independent encoded packets are arriving (up to $k$) and will correspond to a unique full-rank matrix in reduced row echelon form, which can be efficiently computed at the receiver by Gauss-Jordan elimination algorithm.
\end{proof}

\begin{theorem}
Achievable rate of the $\mathbb{F}_q$-linear multicast erasure satellite channel. The reliable achievable rate of the $\mathbb{F}_q$-linear multicast erasure satellite channel is upper bounded by the Max-flow Min-Cut capacity $C$, $R\leq C$. Let the random variable $\lambda$ represent the number of correctly received packets and the random variable $\xi$ represent the rank of the matrix formed by the $\lambda$ packets. The per-multicast link achievable rate can be computed as a function of the coding rate $\rho_i$, $R_i=\rho_i(1-\eta_i)$ with
\begin{eqnarray}
\eta_i=\epsilon_i(Pr(error,\lambda<k) + Pr(error,\lambda\geq k)),
 \label{eq:E3SUB}
\end{eqnarray}

where the first term is a Binomial, $Pr(error,\lambda<k)=\sum_{j=0}^{k-1}{\binom{n}{j}} \left(1-\epsilon_i\right)^{j}\epsilon_i^{n-j}$ while the second term can be approximated by
\begin{eqnarray}
Pr(error,\lambda\geq k)\approx \sum_{j=k}^{n}Pr(\lambda=j)Pr(\xi=k|\lambda=j),
 \label{eq:E4SUB}
\end{eqnarray}

where again the first factor is a Binomial $Pr(\lambda=j)={\binom{n}{j}} \left(1-\epsilon_i\right)^{j}\epsilon_i^{n-j}$, and
\begin{eqnarray}
Pr(\xi=k|\lambda=j)=\prod_{j=0}^{k-1}\left(1-\frac{1}{q^{k-j}}\right).
 \label{eq:E5SUB}
\end{eqnarray}
\end{theorem}

\begin{proof}
The first decoding failure event accounts for the case when the number of correctly received packets at the receiver less than \textit{the required number of packets for successful decoding of} $k$ information packets. The second decoding failure event represents decoding failure due to rank deficiency (i.e. correctly received packets are not linearly independent). Clearly, $\lambda$ is a Binomial random variable, and for practical size fields the expression can be approximated assuming that every packet has the same probability to be innovative. Hence, the only probability that needs to be computed is the conditional probability of the rank to be $k$ conditioned by the number of correctly received packets being greater than $k$. The exact conditional expression in Eq. \ref{eq:E5SUB}, has been derived in \cite{Trullols-Cruces.2011}.
\end{proof}

\vspace{-6pt}
%------------------------------------------------------------------------------------%
\section{NC FUNCTIONAL DOMAIN}
\label{sec:NCFV}
\vspace{-2pt}
\subsection{Integrated NC and NFV architectural design}
When interpreting NC as a functionality to the network, NC function virtualization simply consists in integrating the NC architectural framework described above into existing architectural NFV frameworks. As any virtualized network function (VNF), we will also need first of all to identify physical networks/systems infrastructure with the physical computing, storage, and network resources that will provide NC function virtualization (NCFV) with processing, storage, and connectivity through virtualization, respectively. Note that every virtual infrastructure will have its corresponding time/space scales and communication/computation resources.

The architectural integration is proposed here as a toolbox approach. In particular, a set of elementary NC functionalities are identified that will enable to tailor the use of NC to engineer the throughput and reliability of multiple information network flows depending on the ultimate service operational intent. The resulting VNF will likely operate together with additional virtualized functionalities, which will appropriately be taken care of by the NFV management and orchestration.

Fig. \ref{fig:NCFV_ETSI} shows the architectural design framework of NC and how the functional domain is integrated with the NFV architecture in \cite{NFV002.2013}. Typical solutions for the deployment of NFV are hypervisors and VMs. VNFs can be realized as an application software running on top of a non-virtualized server by means of an operating system (OS) \cite{NFV002.2013}.
\begin{figure*}[ht]
\centering
\includegraphics[scale =0.40] {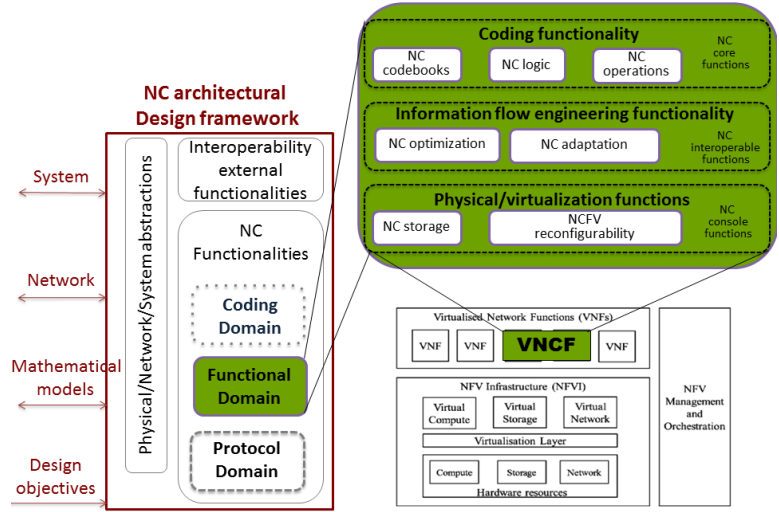}
\caption{Integration of NC architectural design and ETSI NFV architecture in \cite{NFV002.2013} as a toolbox of NC design domains.}
\label{fig:NCFV_ETSI}
\end{figure*}

\subsection{NC Functional Domain design}
Based on the key features of NFV design, our proposed set of basic NC functionalities are distributed into three hierarchical levels based on their significance, universality, and availability. In particular, we define common elementary functionalities \cite{Tan.2016} as follows:
\begin{itemize}
	\item \textbf{NC coding functionalities }
\begin{itemize}
	\item NC Logic: logical interpretation of coding use, coding scheme selection (intra-session/inter-session, coherent/incoherent, file-transfer/streaming, systematic/non-systematic), coding coefficients selection (random/deterministic), etc. 
	\item NC Coding: elementary encoding/re-encoding/decoding operations, encapsulation/de-encapsulation, adding/removing headers, etc. 
	\item NC operations: computation and storage optimization.
\end{itemize}

\item \textbf{NC information flow engineering functionalities }
\begin{itemize}
	\item NC optimization: resource allocation and optimal allocation of NC parameters possibly differentiating network flows according to design targets and (statistical) status of the networks (congestion, link failures, etc).
	\item NC adaptation: control of the time scales for optimization and corresponding reconfiguration across re-encoding nodes across the network. 
\end{itemize}

\item \textbf{NC physical/virtualization functions }
\begin{itemize}
	\item NC Storage: interactions with physical memory. 
	\item NC Feedback Manager: settings and interaction for acquisition of network statistics feedback. 
	\item NC Signaling: coordinating signaling parameters and involved encapsulations and/or protocols. 
\end{itemize}
\end{itemize}

Note that our proposed set of basic functionalities can be easily increased with specialized functionalities and/or future-purpose functionalities.

\vspace{-6pt}
%------------------------------------------------------------------------------------%
\section{Use Case: Geo-Controlled network connectivity} 
\label{sec:CaseStudy}
\vspace{-2pt}
\subsection{Physical system/network abstraction}
Our interest is to validate the performance of a specific design using the functionalities proposed in the previous sections. 
In particular, we propose a specific VGNCF design for reliable communication services over satellite, called \textit{geo-controlled network connectivity}. 
In the illustrative design proposed here, the optimization functionality will use databases with geo-tagged link statistics and geo-location information of network nodes in the service coverage area for some service operational intents. 
The overall design target here is to achieve a given connectivity and/or target reliability throughout the service coverage area given constraints of computational complexity/energy efficiency.

The interest of NFV is that the same virtualized network functionality can be applied over different operational frameworks and services as well as over different underlying physical networks, including satellite or hybrid networks thus enabling softwarization and rapid innovative deployment. 
In Fig. \ref{fig:VGNCF_Architecture}, we abstract the underlying physical network/system to identify functionalities of NC. All network nodes in the deployment area are geo-localized based on GNSS network. Our proposed use case applies to multiple scenarios consisting of terminals (randomly) distributed in the service area both within and beyond cell coverage of satellites. Connectivity to/from the source is realized via satellites and/or transport networks (e.g. internet, cellular networks), depending on availability but in all cases they are abstracted for our design as sender/sink (end-user) nodes, respectively, through a virtualization layer. %Receivers are then abstracted accordingly at the (sub)layer where NC software will run. 
Logical network then illustrates information flow functionality of network nodes including encoding, re-encoding, and decoding points.  
\begin{figure*}
     \centering
     \includegraphics[scale =0.72]{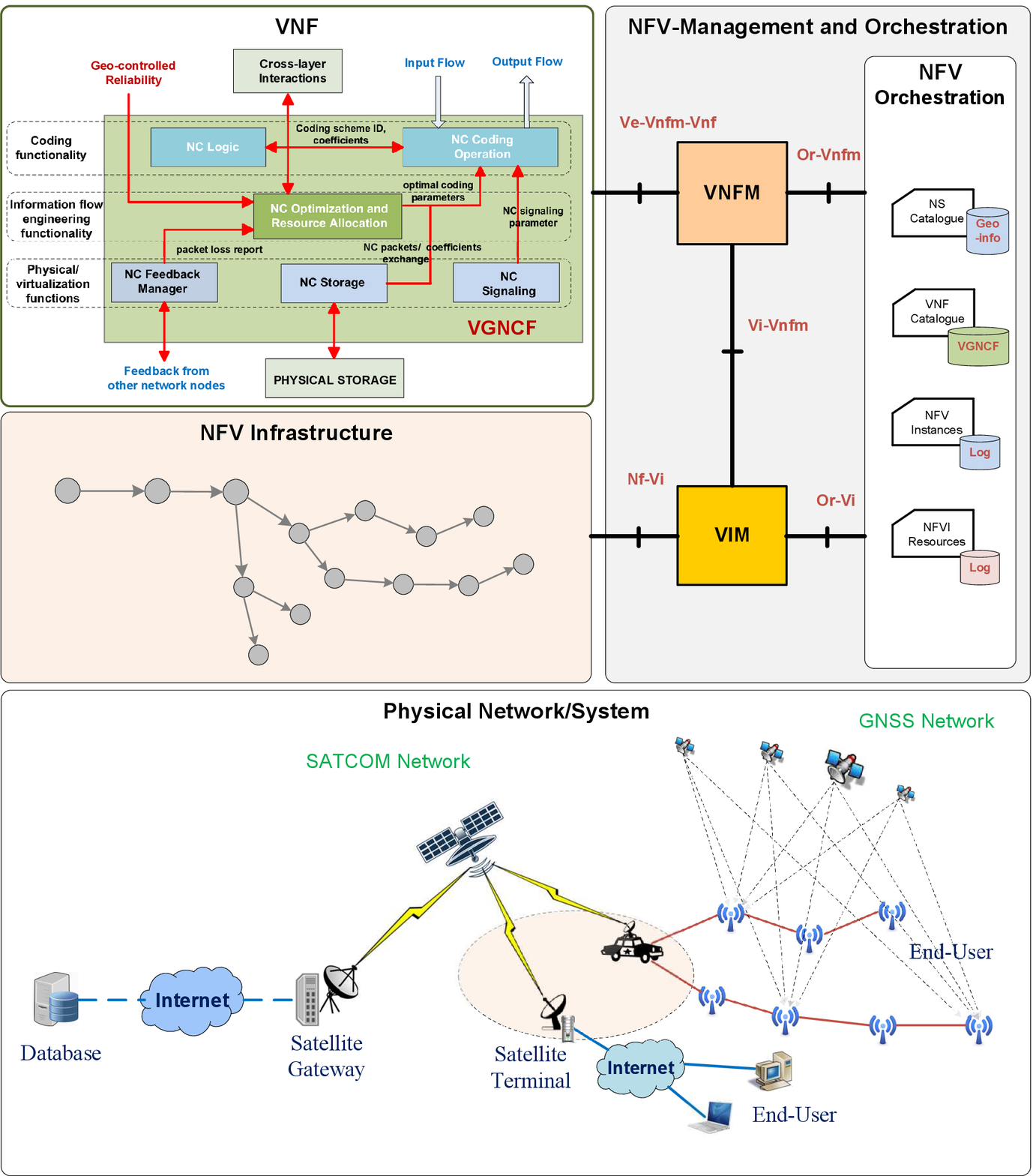}
     \caption{VGNCF integrated as a VNF in ETSI NFV architecture.}
     \label{fig:VGNCF_Architecture}
\end{figure*}

\subsection{Virtualized Geo-NC function architecture}
Fig. \ref{fig:VGNCF_Architecture} denotes our proposed virtualized geo-NC function (VGNCF) in which we show how NCF can be integrated with the ETSI NFV architecture given the abstracted underlying physical system/network as part of NFV Infrastructure (NFVI). The design also indicates exchanges between VGNCF and NFV-Management and Orchestration (NFV-MANO) over reference points. 
Specifically, NFV-MANO is the grey block in Fig. \ref{fig:VGNCF_Architecture} that includes three functional blocks: NFV Orchestrator (NFVO), VNF Manager (VNFM), and Virtualized Infrastructure Manager (VIM). Where NFVO block is considered as the brain of NFV architecture which has two main responsibilities: (1) the orchestration of NFVI resources across multiple VIMs and (2) the life-cycle management of all network services. 
While VNFM manages the life-cycle of VNF instances, VIM is responsible for managing and controlling NFVI resources including physical, virtual, and software resources. 
Further details of the NFV-MANO architectural framework can be found in \cite{NFV-MAN001.2014}.

The detailed functional domain of VGNCF is shown in the VNF block of Fig. \ref{fig:VGNCF_Architecture}.
At the NC core blocks, interactions with other nodes bring into agreement on coding schemes, coefficients selection, etc. NC coding block receives all inputs such as coding scheme, coefficients, coding parameters, packets from storage block, and signaling to perform elementary encoding/re-encoding/decoding operations, etc. At the NC interoperable blocks, geographical location-based information and level of reliability provided by \textit{geo-controlled reliability} block will be given to the NC optimization and resource allocation block so that optimal coding parameters are generated to the NC coding operation block. In addition, packet loss feedback from other network nodes is also an important factor in the resource allocation process. At the last stage, NC console blocks, which connect directly to physical storage and feedback from other network nodes to provide information packets and packet loss rate to the upper stage, respectively.
The proposed NCF will interact with protocol domain, i.e. network, at a suitable position in reference model of each standardized protocol stack e.g. transportation layer.

As denoted in Fig. \ref{fig:VGNCF_Architecture}, NFV-MANO needs the repositories that hold different information regarding network services (NSs) and VNFs (VGNCF is part of VNFs). There are four types of repositories as follows
\begin{itemize}
	\item \textbf{VNF catalogue} represents the repository of all usable VNF packages, supporting the creation and management of the VNF packages. %It is mainly used by VNFM block during the instantiation and lifecycle management of a VNF instance. %However, NFVO also can use VNF catalogue to manage and orchestrate NSs and virtualized resources on NFVI.
	\item \textbf{NS catalogue} represents the repository of all usable NSs.
	\item \textbf{NFV instances} is the repository that holds details of all VNF instances and NS instances, represented by either a VNF record or a NS record, respectively, during the execution of VNF/NS life-cycle management operations.
	\item \textbf{NFVI resources} is the repository that holds information about NFVI resources utilized for the establishment of NS and VNF instances.
\end{itemize}

NFV-MANO functional blocks exchange information with VNFs (e.g. our proposed VGNCF) and between them via reference points, which are summarized as follows:
\begin{itemize}
	\item \textbf{Ve-Vnfm-Vnf}: this reference point is used for exchanges between VNFs and VNFM. The goal of Ve-Vnfm-Vnf includes VNF instantiation/update/termination, verification of VNFs' operating state, forwarding of configuration and events from VNF to VNFM and vice versa, etc. Especially, this reference point also provides \textit{geo-information} such as geo-tagged link statistics and geo-location information for the optimization functionality of network codes.
	\item \textbf{Nf-Vi}: this reference point is used for exchanges between VIM and NFVI by supporting allocation VM with indication of compute/storage resource, updating VM resource allocation, migrating/terminating VM, forwarding of configuration information, failure events regarding NFVI to VIM, etc.
	\item \textbf{Vi-Vnfm}: this reference point is used for exchanges between VNFM and VIM. The exchanges over Vi-Vnfm include NFVI resources reservation information retrieval, NFVI resources allocation/release, etc. %configuration information between reference point peers, %and forwarding to the VNFM information such as events, measurement results, and usage records regarding NFVI resources used by VNFs.
	\item \textbf{Or-Vnfm}: this reference point is used for exchanges between NFVO and VNFM with main tasks such as NFVI resources authorization/validation/reservation/release/resources allocation for a VNF, VNF instantiation, VNF instance query/update/termination, forwarding of events, other state information about VNF, etc.
	\item \textbf{Or-Vi}: this reference point is used for exchanges between NFVO and VIM. Main tasks of Or-Vi include NFVI resource reservation/release, NFVI resource allocation/update, forwarding of configuration information, usage records regarding NFVI resources to NFVO, etc.
\end{itemize}

\subsection{Exchanges between VGNCF and NFV-MANO blocks}
We show the detailed exchanges between our proposed VGNCF and NFV-MANO blocks via reference points as defined by ETSI reference architecture. As denoted in Fig. \ref{fig:VGNCF_EXCHANGES}, the management and orchestration of VGNCF is highlighted by the following points:
\begin{enumerate}
	\item NS catalogue holds information of all usable NSs in terms of VNFs and description of their connectivity through virtual links. Geo-information e.g. geo-tagged link statistics and geo-location information is also stored in NS catalogue.
	\item VNF catalogue holds information of all usable VNFs in terms of VNF Descriptors (VNFDs). The proposed VGNCF is part of this repository.
	\item NFV instances repository keeps record from VGNCF during its execution of management operations and life-cycle management operations.
	\item NFVI resources hold information about NFVI resources utilized for VNF/NS operations including VGNCF.
	\item Performance monitoring and analytics blocks take charge of collecting and analyzing non-real-time/real-time data e.g. loss rate resulting from operation of VGNCF so that Orchestration Engine (OE) can follow and update coding policy.
	\item Policy management receives information from performance monitoring and analytics block to decide whether coding policy should be changed or not. Some relevant issues can be identification of Point-of-Presence, selection of coding points, resource allocation, computational complexity constraints, etc.
	\item Interaction between VNFs should be considered. For instance, a virtualized routing function implemented provides route information to VGNCF so that our proposed VGNCF can perform the optimization functionality.
\end{enumerate}

\begin{figure*}
     \centering
     \includegraphics[scale =0.60]{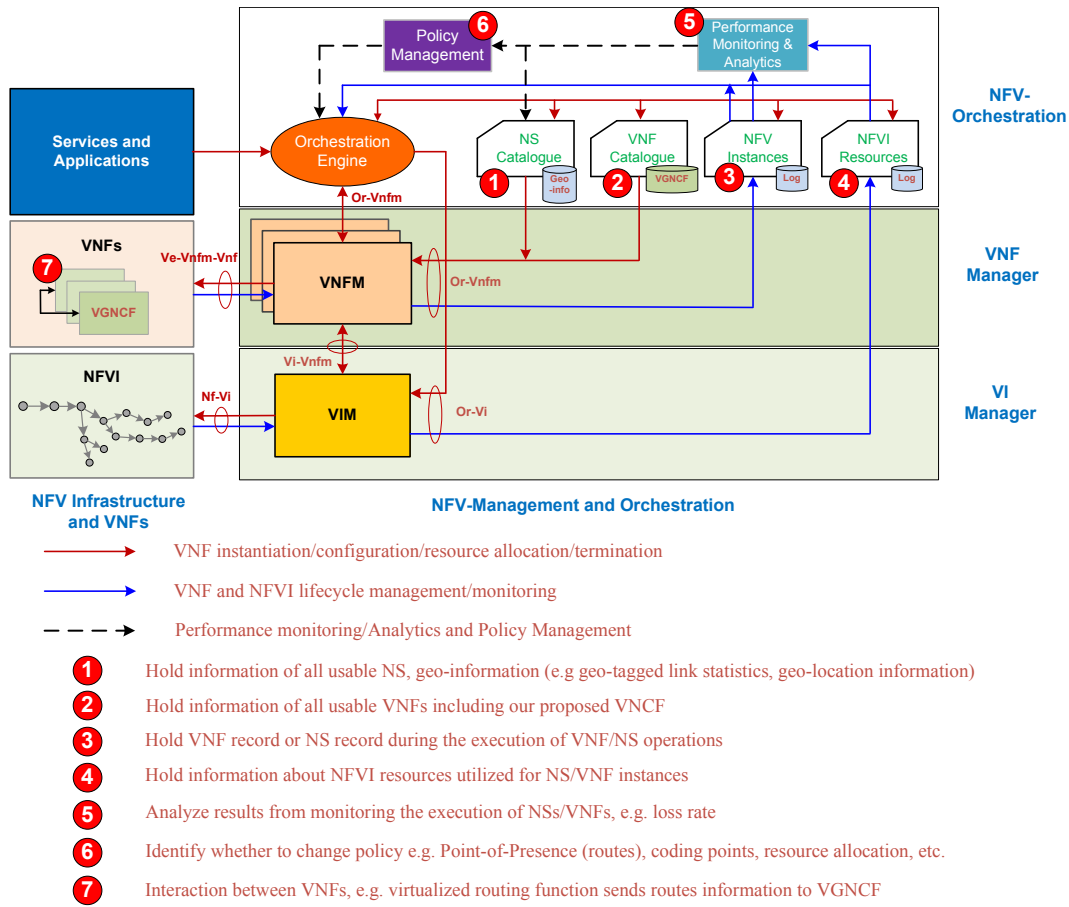}
     \caption{Exchanges between VGNCF and NFV-MANO via reference points.}
     \label{fig:VGNCF_EXCHANGES}
\end{figure*}

It is noted that the design of efficient algorithms for performance monitoring/analytics and policy management is a promising challenge direction to enforce into NFV-Orchestration of not only NCFV but also general NFV design. Therefore, such design has not been investigated in the present work and remains as part of future work.
%%%%%%%%%%%%%%%%%%%%%%%%%%%%%%%%%%%%%%%%%%%%%%%%%%%%%%%%%%%%%%%%%%%%%
%\subsubsection{General procedure for NFV-Orchestration}
\begin{figure*}
     \centering
     \includegraphics[scale =0.50]{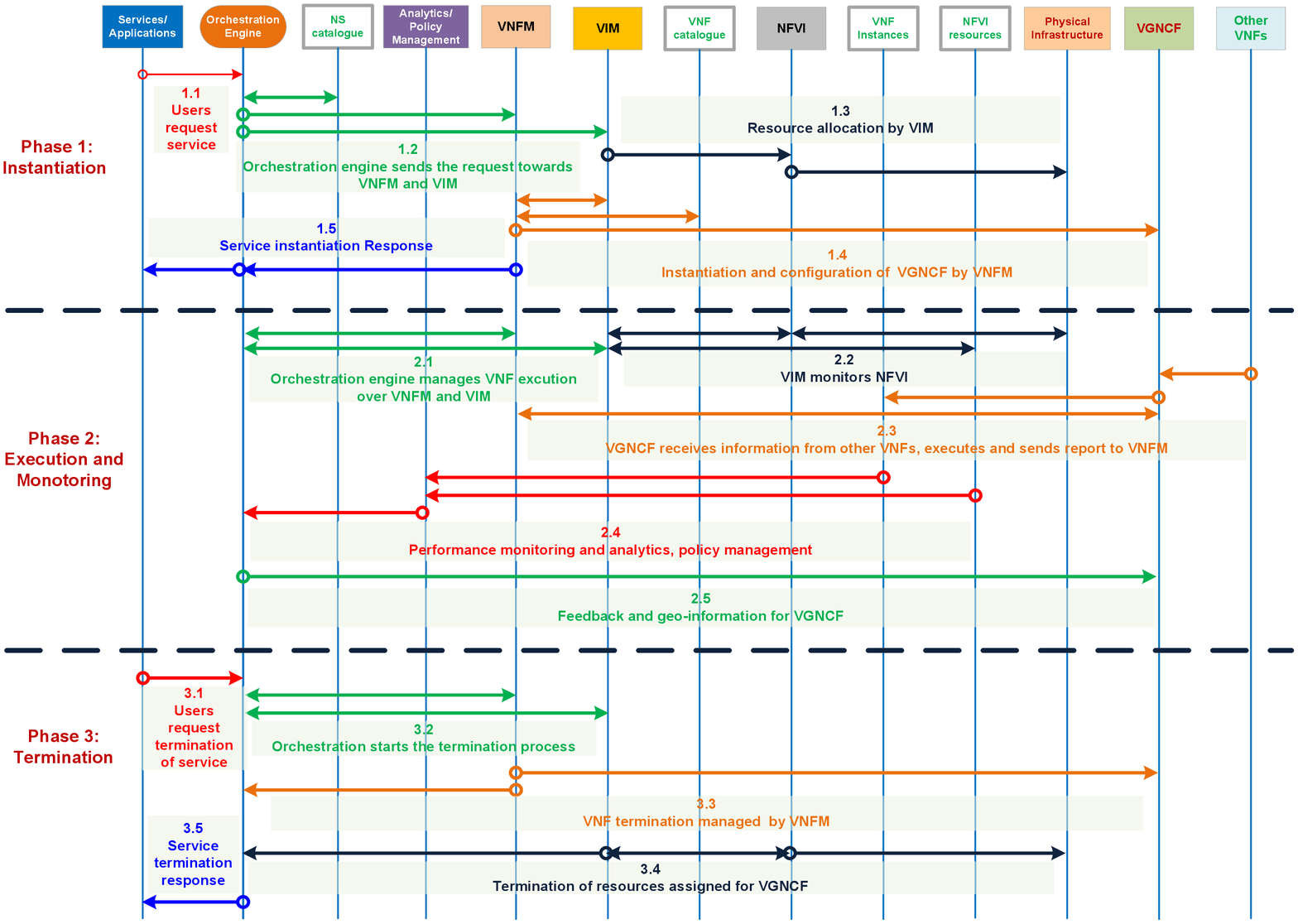}
     \caption{Procedure for the instantiation, execution and monitoring, and termination of VGNCF.}
     \label{fig:PROCEDURE}
\end{figure*}

In Fig. \ref{fig:PROCEDURE}, we denote a general procedure of the instantiation, performance monitoring, and termination corresponding life-cycle management of VGNCF execution in the case that reliability functionality with NC is activated by the centralized controller using our proposed architectural design. The process consists of three phases, including instantiation, execution and monitoring, and termination.
The details are introduced in the following:
\begin{itemize}
	% Phase 1
	\item \textit{Phase 1.1}: customers/users request the reliability functionality for a communication service over satellite provided by VGNCF.
	\item \textit{Phase 1.2}: the request requires the OE has to read the NS catalogue for the description of NSs and geo-information according to connectivity from the source to the destination node. Based on the descriptor received from the NS catalogue, the engine then delivers a message towards the VNFM and VIM to request for the instantiation, configuration of VGNCF, and resource allocation, respectively.
	\item \textit{Phase 1.3}: VIM takes charge of managing and allocating resource required for the instantiation and execution of VGNCF via VNFI. The underlying physical infrastructure would be allocated respectively.
	\item \textit{Phase 1.4}: on the other hand, VNFM interacts with VNF catalogue for the description of VGNCF in terms of its deployment and operational behavior requirements and delivers necessary configuration for the instantiation of VGNCF at coding points. 
	\item \textit{Phase 1.5}: the instantiation phase has been done by acknowledgment from VNFM to OE and the user in order to confirm the successful deployment of VGNCF on coding points.
	% Phase 2
	\item \textit{Phase 2.1}: this phase describes the interaction of our proposed VGNCF and the management and orchestration system. During the execution of VNF/NS life-cycle management operation, the OE keeps interacting with VNFM and VIM for real-time management and monitoring.
	\item \textit{Phase 2.2}: the VIM takes care of NFVI and network resources while current information of NFVI resources utilized for VNF/NS operations is stored in NFVI resources catalogue for future use and also available for using by the OE.
	\item \textit{Phase 2.3}: the VGNCF receives routing information from another VNF, e.g. virtualized routing function, and configuration from VNFM for executing the coding functionality. Note that information from operation of VGNCF should be recorded in VNF instances repository for performance monitoring.
	\item \textit{Phase 2.4}: the performance monitoring/analytics and policy management blocks gather information from real-time/non real-time operation of VGNCF recorded in NFV instances and NFVI resources for analytical framework. The output is then sent to the OE in order to update configuration and coding policy.
	\item \textit{Phase 2.5}: the VGNCF requires updating feedback and geo-information for optimization functionality.
	%Phase 3
	\item \textit{Phase 3.1}: assume that users would like to deactivate NC functionality. Then, a request is sent to the OE.
	\item \textit{Phase 3.2}: the OE delivers the termination message for the VNFM and VIM.
	\item \textit{Phase 3.3}: the VNFM deactivates all functionalities of VGNCF.
	\item \textit{Phase 3.4}: corresponding resources provided for VGNCF will be deallocated and released it back to the VIM. 
	\item \textit{Phase 3.5}: after receiving the termination response from both VNFM and VIM, the OE acknowledges completion of the termination phase to the central controller/users.
\end{itemize}

\subsection{Design of Coding Functionality}
Let $\epsilon_{i}$ be the erasure rate of each link $i$ and $\overline{\epsilon}$ be the vector of per-link erasure rates e.g. $\overline{\epsilon}=(\epsilon_{1},\epsilon_{2})$ for $2$ hops. If all links are equal, we use a unique value $\epsilon$. A source generates $n$ coded packets from a set of $k$ information packets over linear map by a coefficient matrix generated using coefficients from the same Galois field of size $g=2^{q}$. Then, coding rate is given by $r=k/n$. Assume that packet length is $M$ bits, resulting in $m=M/q$ symbols per packet.

Let $\eta_{i}\left(r,\epsilon_{i}\right)$ denote the residual erasure rate after decoding at each single hop $i$. 
The NC codebook will be systematic as introduced in Sec. \ref{sec:SNC}.

In view of multi-hop line networks, we define the reliability after decoding at hop $h$ as 
\begin{eqnarray}
\rho_{R}^{NC}(r,\overline{\epsilon},h)=\text{\ensuremath{\prod}}_{i=1}^{h}\left(1-\eta_{i}\left(r,\epsilon_{i}\right)\right).
 \label{eq:E1}
\end{eqnarray}

Accordingly, we also define the reliability at hop $h$ for the uncoded case as 
\begin{eqnarray}
\rho_{R}^{noNC}(\overline{\epsilon},h)=\text{\ensuremath{\prod}}_{i=1}^{h}\left(1-\epsilon_{i}\right). 
\label{eq:E2}
\end{eqnarray}

\subsection{Optimization Functionality}
\label{sec:OPT}
\subsubsection{Computational complexity:}
\label{sec:COMP}
Let $\beta_{0}$ denote the limitation on computational complexity of each node implementing VGNCF. Let $\beta_{S}(r)$, $\beta_{R_{j}}(r)$, and $\beta_{D}(r)$ denote computational complexity required for NCF at the source, re-encoding points $R_{j} (1\leq j<h)$, and the destination, respectively.
\begin{proposition}
The number of multiplications and additions required for encoding process is $N_{enc}^{M}=(n-k)km$ and $N_{enc}^{A}=(n-k)(k-1)m$, respectively. Therefore, computational complexity for encoding in terms of multiplications and additions is $\beta^{enc}(r)=N_{enc}^{M}+N_{enc}^{A}=(n-k)m(2k-1)$. Whereas computational complexity of finite-length decoding complexity using Gaussian Elimination algorithm in terms of the number of multiplications and additions in $GF(2^{q})$, denoted by $\beta^{dec}(r)=N_{dec}^{M}+N_{dec}^{A}$, can be found in \cite{Garrammone.2013}. 
\end{proposition}

Each re-encoding point $j$ will decode and re-encode the linear combinations before forwarding the coded packets towards next hops. This implies that each re-encoder is also a receiver on the path. Moreover, without decoding the re-encoder does not know which packets are innovative, while with decoding the relay may do a more intelligent re-encoding operation. Therefore, the complexity required for the relay is total of computational complexity for decoding and re-encoding, $\beta_{R_{j}}(r)=\beta^{dec}(r)+\beta^{enc}(r)$. Whereas $\beta_{S}(r)=\beta^{enc}(r)$ and $\beta_{D}(r)=\beta^{dec}(r)$.
Assume that coding rate $r$ is the same for all the nodes. Then $\beta_{S}(r)$ is less than $\beta_{R_{j}}(r)$ and $\beta_{D}(r)$. 

Note that computational complexity can be equivalently mapped into energy consumption in terms of power per computational unit e.g. logic gates. However, such specification depends on specific hardware that NC function is deployed. As a generalized approach, we thus briefly mention energy consumption via computational complexity in terms of multiplications and additions. An example of the energy consumption of NC applied for low-power electronics can be found in \cite{Angelopoulos.2011}.

\subsubsection{Utility Function:}
Identification of utilities allows the optimization of the network coded flows. We are interested in the design of when NC should be activated in terms of (1) computational complexity/energy consumption and (2) target reliability after decoding $(\rho_{0})$. In order to do so, we define the following utility function: 
\begin{eqnarray}
u^{act}(r,\overline{\epsilon},\rho_{0})=\frac{f^{NC}(r,\overline{\epsilon},\rho_{0})}{f^{COST}(r)},
 \label{eq:E3}
\end{eqnarray}

where $f^{NC}(r,\overline{\epsilon},\rho_{0})$ accounts for the goodness of the encoding/decoding scheme in achieving target performance $\rho_{0}$. Whereas $f^{COST}(r)$ accounts for the cost in terms of computational complexity/energy consumption. The utility function denotes the energy efficiency in terms of the ratio between the goodness and the computational complexity. 
We define the goodness and cost function, respectively as follows 
\begin{eqnarray}
f^{NC}(r,\overline{\epsilon},\rho_{0})=\rho_{R}^{NC}(r,\overline{\epsilon})-\rho_{0},
 \label{eq:E4}
\end{eqnarray}
\begin{eqnarray}
f^{COST}(r)=\beta_{S}\left(r\right),
 \label{eq:E5}
\end{eqnarray}

with $\beta_{S}\left(r\right)$ is computational complexity of the source as defined in Section \ref{sec:COMP}.

\subsubsection{Optimized operative ranges of performance:}
The source identifies the upper bound for the energy efficiency of coding strategy and the optimized $\rho_{R}^{NC}(.)$ that the source should provide to the destination according to a given $(\overline{\epsilon},\rho_{0})$ and constraint of computational complexity. Note that given the constraint, our virtual reliability functionality should self-adapt to computational limitations. Underlying geo-tagged channel statistics are assumed to be stored in NS catalogue (see Fig. \ref{fig:VGNCF_Architecture}).
The upper bound for energy-efficient coding rate is given by the following proposition.
\begin{proposition}
The source identifies at which rate it maximizes its own utility under constraints of computational complexity/energy consumption by the following optimization strategy:
\begin{equation}
\begin{aligned}
& \underset{r}{\text{argmax}}
& & u^{act}(r,\overline{\epsilon},\rho_{0}) \\
& \text{subject to}
& & \beta_{D}(r)\le\beta_{0},\\
&&& \beta_{R_j}(r)\le\beta_{0}.
\end{aligned}
\label{eq:E6}
\end{equation}
Numerical results reveal that the utility function has the property of quasi-concavity. 
Moreover, $\beta_{D}(r)$ and $\beta_{R_j}(r)$ are increasing functions with redundant coded packets. Therefore, Problem (\ref{eq:E6}) is equivalently quasi-convex optimization and can be efficiently solved by bisection methods \cite{Boyd.2004}.
The utility may hold a minus value which represents the penalty since the design target is not satisfied. 
Assume computational complexity limitation is large enough, the optimal utility is not necessarily at coding rates that make $\rho_{R}^{NC}(.)=1$.
\end{proposition}

Our proposed strategy identifies optimal points that bring optimal benefit for the source's viewpoint. It is thus necessary to identify optimized operative ranges of performance so that the destination is aware and admits some variations in the quality. 
At the source or a centralized controller, the cognitive algorithm to identify optimized operative ranges is briefly realized as follows:
 	
	(1) identify maximal utility $u_{max}^{act}(r,\overline{\epsilon},\rho_{0})$ and respective $r$, $\rho_{R}^{NC}(r,\overline{\epsilon})$ given $(\overline{\epsilon},\rho_{0})$,	
	
	(2) determine $r$ and $\rho_{R}^{NC}(r,\overline{\epsilon})$ that satisfies $u_{min}^{act}(.) \leq u^{act}(r,\overline{\epsilon},\rho_{0})\leq u_{max}^{act}(.)$, with $u_{min}^{act}(.)$ is the lower bound, and
	
	(3) activate NC functionality if the range of performance is acceptable by users. 

	Note that we can assume the lower bound is just the smallest value of utility that the target $\rho_{0}$ is still satisfied. For the sack of comparison, we also assume that even though the target cannot be met for any coding rate due to significant erasure process and/or low complexity constraint, the source should show its effort to improve network reliability as much as possible regardless its utility may still be minus.	

%\subsection{NC design protocol domain}
%The proposed NC functional domain should interact with protocol domain i.e. network. We thus consider an instantiation in the existing LTE and WiFi protocol stacks. We first need to identify a suitable position in the reference model of each standardized protocol stack for NC sublayer. Towards the compliance with the upcoming wireless systems e.g. 5G, our proposed NC functionality should be adapted at IP layer. Depending on the functionalities of a node in the network, i.e. either source, sink, or intermediate relay, we identify three main functional paths of encoding/decoding/re-encoding operations with respect to data flow’s direction
%------------------------------------------------------------------------------------%
\section{PERFORMANCE EVALUATION}
\label{sec:Numerical}
\vspace{-2pt}
\subsection{Per-link achievable rate region}
\label{sec:RateRegion}
In this section, we indicate that by adding redundancy not only at the source but also at intermediate nodes, per-link achievable rate region goes beyond that of the case that NC is applied only at the source.
We consider a single-source multicast network as presented in many communication systems, e.g. a satellite network in which a satellite multicasts information packets to a number of receivers within the service coverage area \cite{MAVazquez.2015Fq}. Where both the uplink and the downlink are erasure channels. 
We denote the first case as \textit{NC-case} where NCF is applied at both the source and the intermediate node. While the second one is with NC at the source only and known in the NC literature as \textit{end-to-end coding} to which NC can be compared with since the erasure seen by the encoder is the total erasure in the network. 

It is noted that we use interchangeably the target reliability $\rho_{0}$ and $\eta_0$, the target residual packet erasure rate after decoding at the receiver, with $\eta_0 = 1 - \rho_0$.
Let denote $R^{NC}=r\left(1-\eta^{NC}\right)$ and $R^{e2e}=r\left(1-\eta^{e2e}\right)$ as the achievable rate for NC case and end-to-end coding, respectively, where $\eta^{NC}$ is residual erasure rate after decoding for 2-hop networks and $\eta^{e2e}$ corresponds to residual erasure rate after decoding for single-hop networks but with erasure rate is total erasure of the 2-hop network. 
In the below figures, we show achievable rate region with respect to per-link erasure rates for different target $\eta_{0}$. First, we consider and analyze the behavior of $R^{e2e}$  and $R^{NC}$ with a limitation of coding rate. We then evaluate the impact of different limitations of coding rate on the behavior of $R^{e2e}$ and $R^{NC}$.

\begin{figure*}
     \centering
     \subfloat[][$R^{e2e}$]			{\includegraphics[scale =0.50]{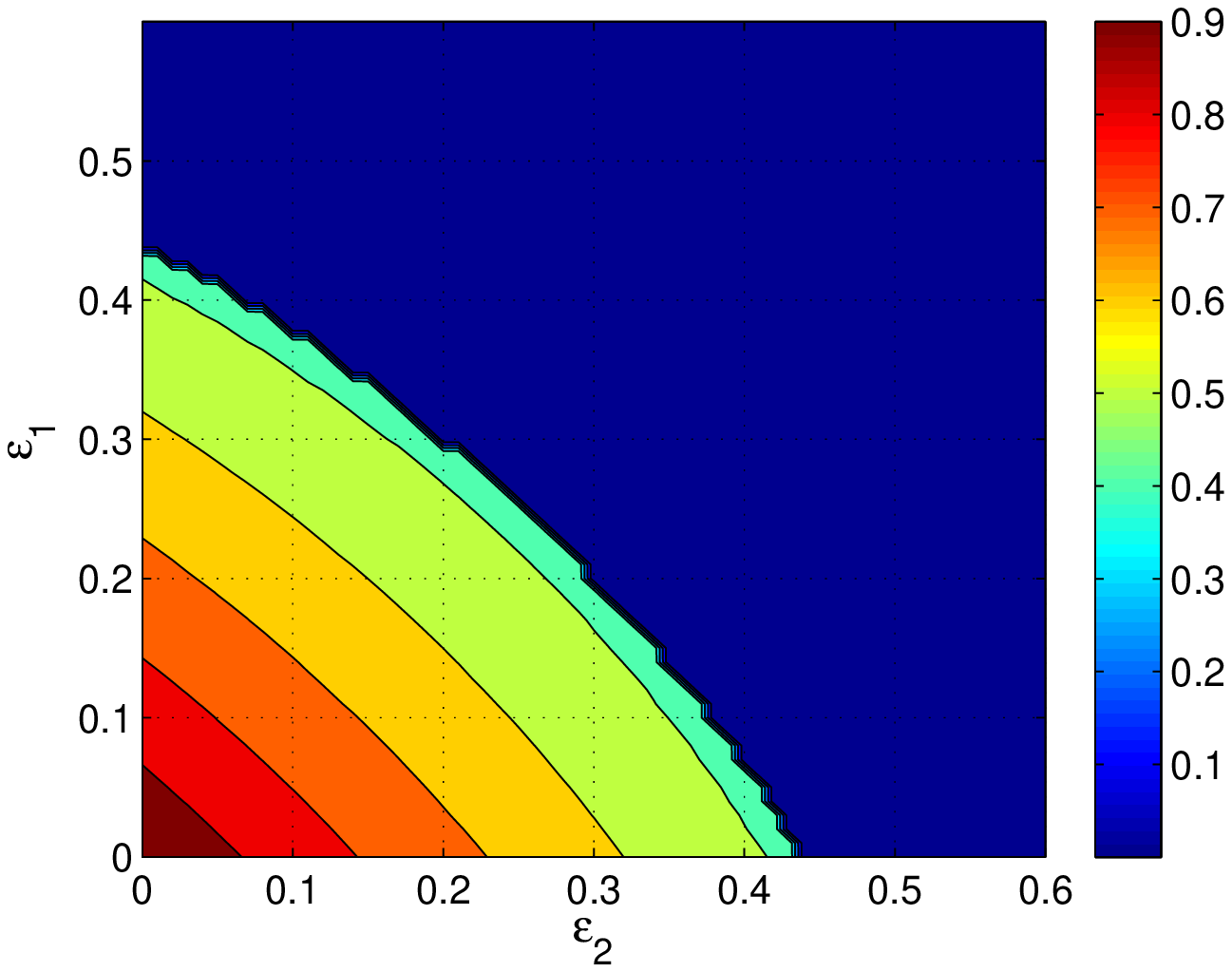}\label{fig:Re2e_5_r_05}}
     \subfloat[][$R^{NC}$]			{\includegraphics[scale =0.50]{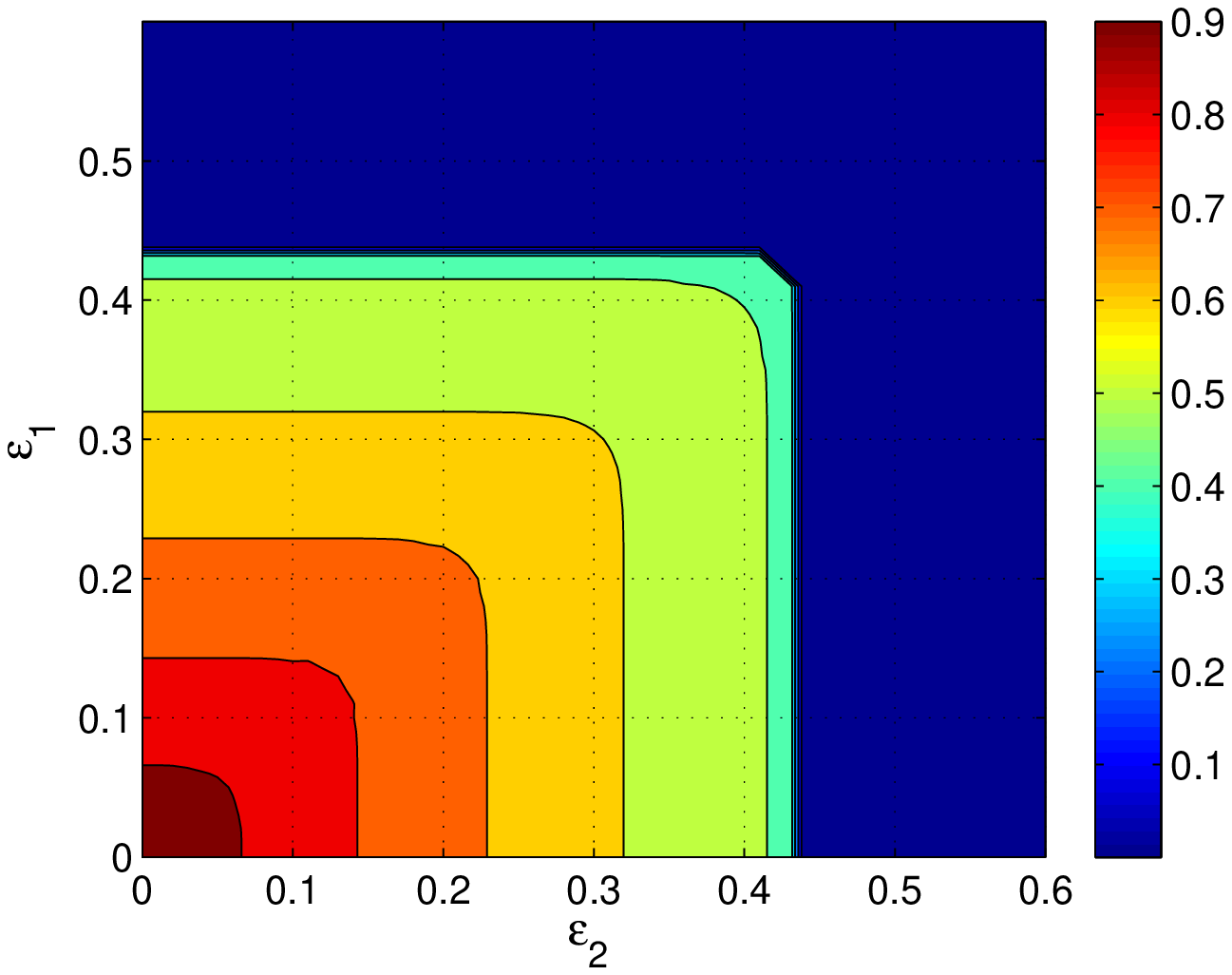}\label{fig:Rnc_5_r_05}}
     \caption{Per-link achievable rate region with respect to link erasure rates for $\eta_{0}=5\%$, $r\in[0.5,1]$.}
     \label{fig:R_5_r_05}
\end{figure*}
\begin{figure*}
     \centering
     \subfloat[][$R^{e2e}$]			{\includegraphics[scale =0.50]{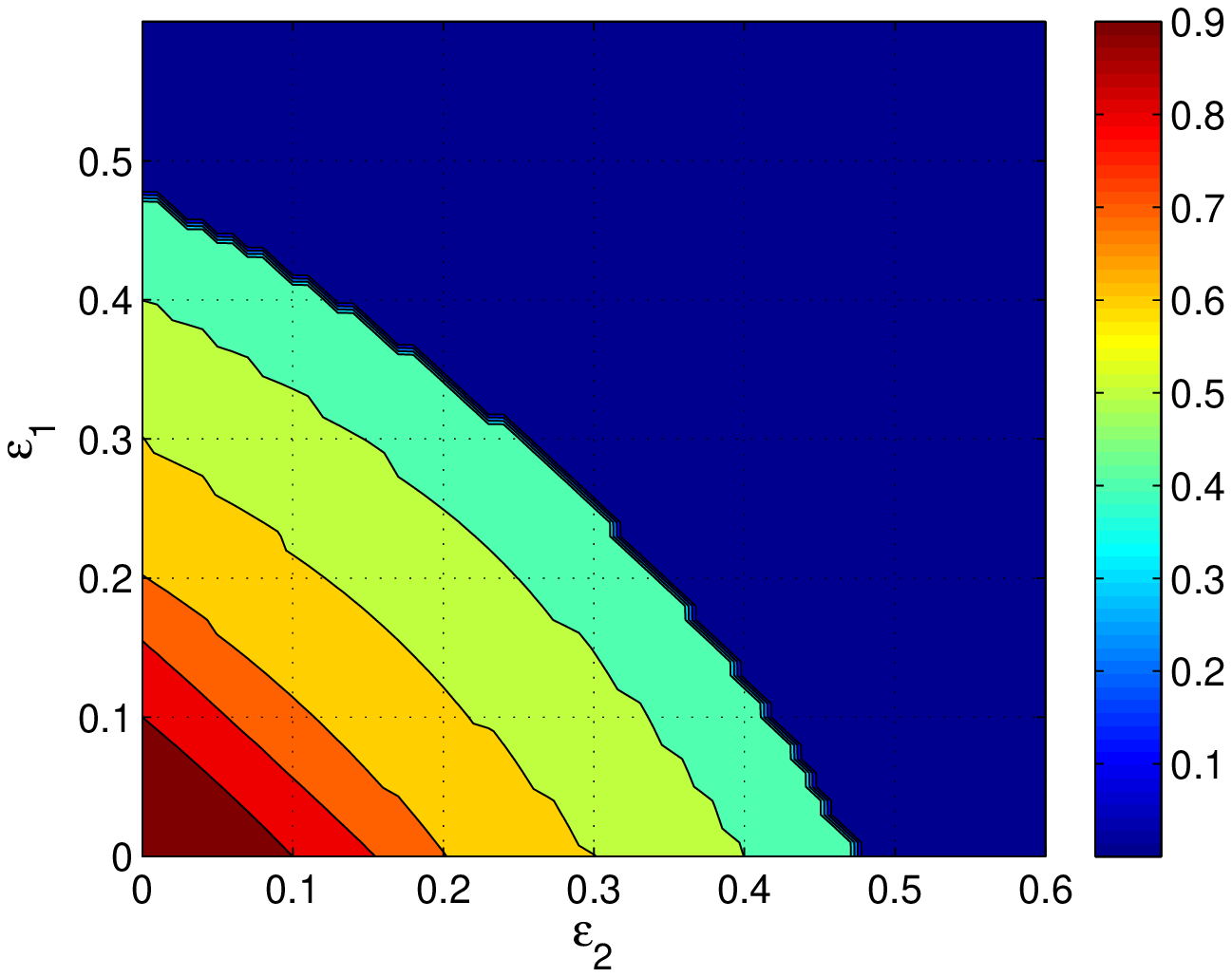}\label{fig:Re2e_15_r_05}}
     \subfloat[][$R^{NC}$]			{\includegraphics[scale =0.50]{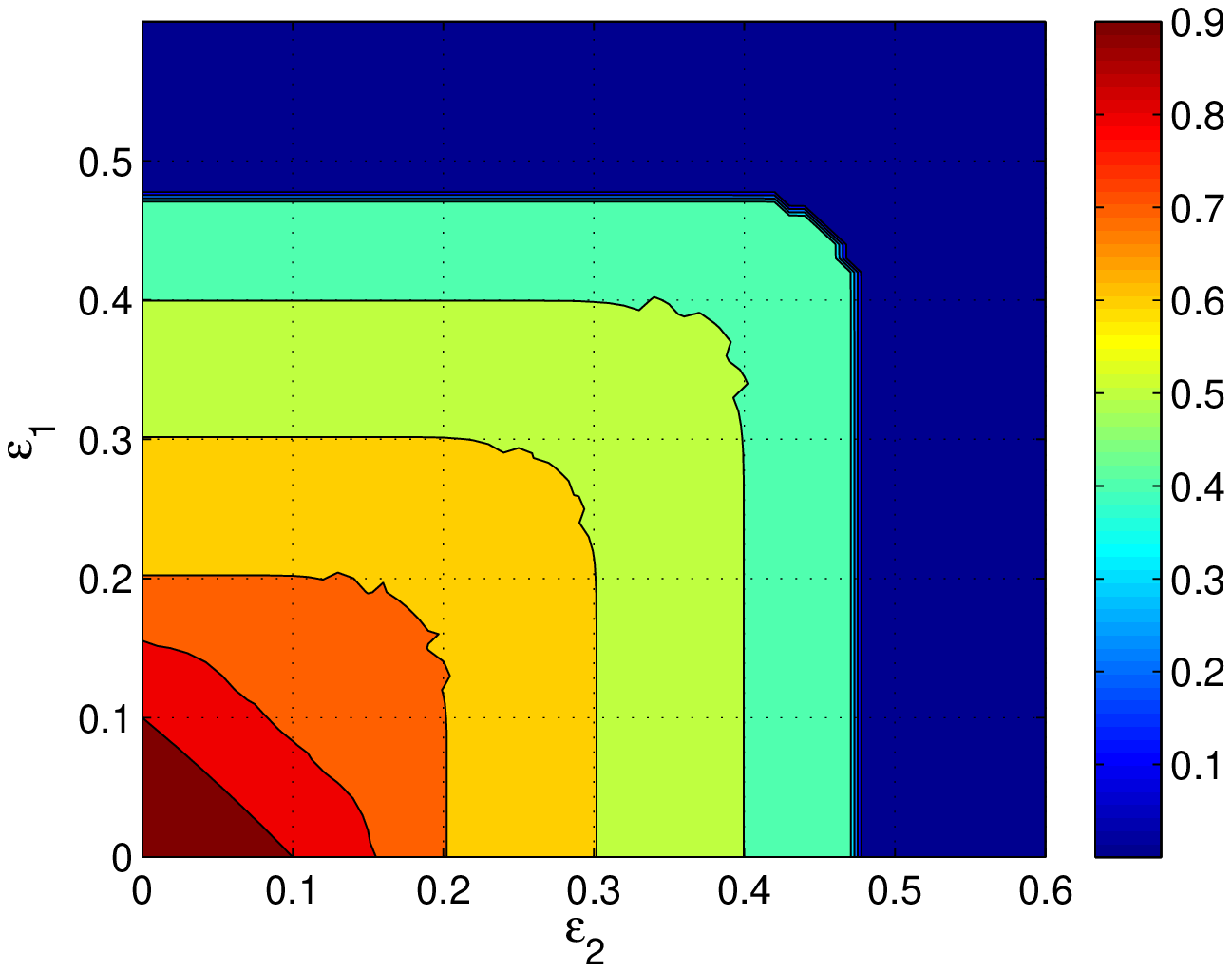}\label{fig:Rnc_15_r_05}}
     \caption{Per-link achievable rate region with respect to link erasure rates for $\eta_{0}=15\%$, $r\in[0.5,1]$.}
     \label{fig:R_15_r_05}
\end{figure*}
\begin{figure*}
     \centering
     \subfloat[][$R^{e2e}$]			{\includegraphics[scale =0.50]{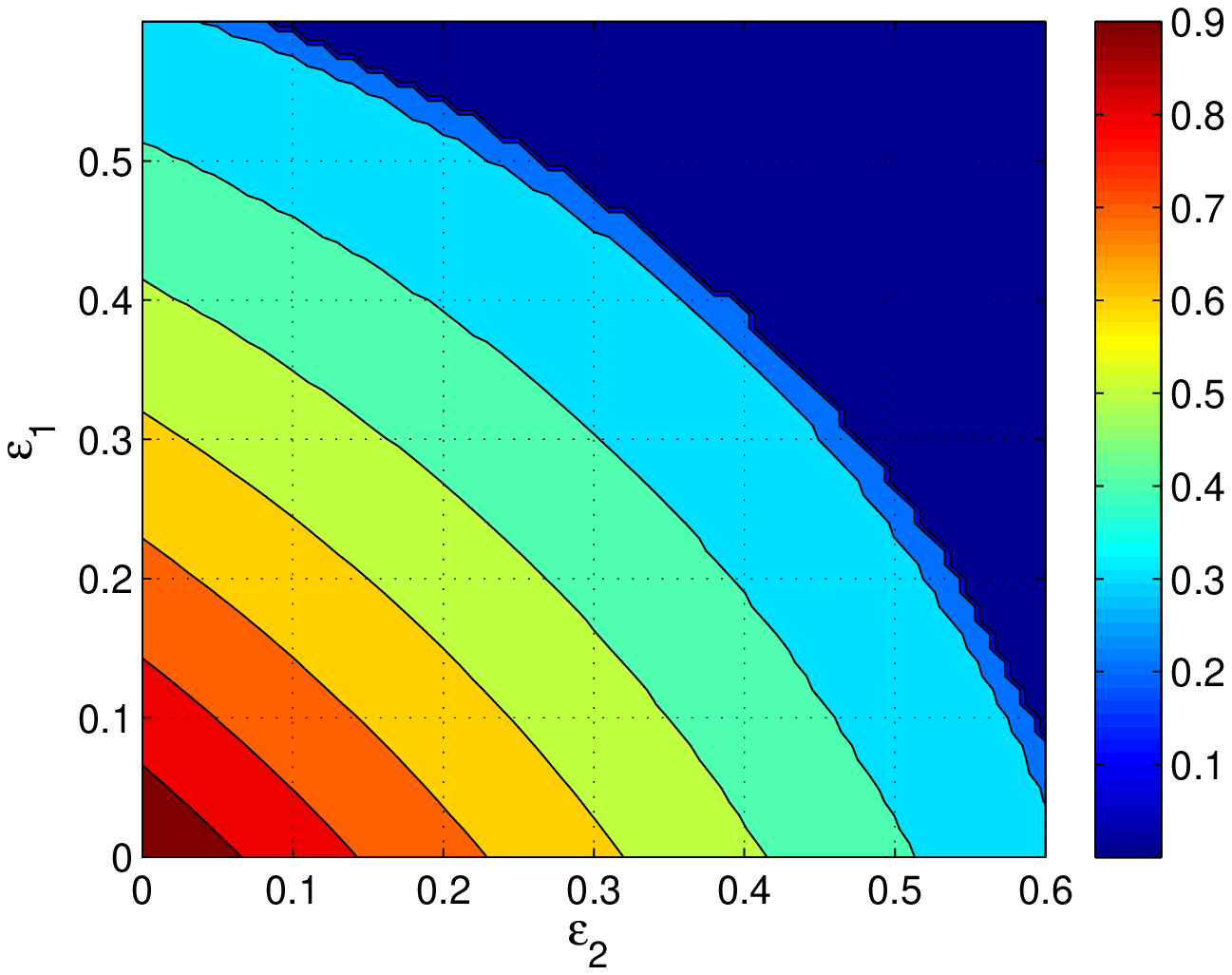}\label{fig:Re2e_5_r_03}}
     \subfloat[][$R^{NC}$]			{\includegraphics[scale =0.50]{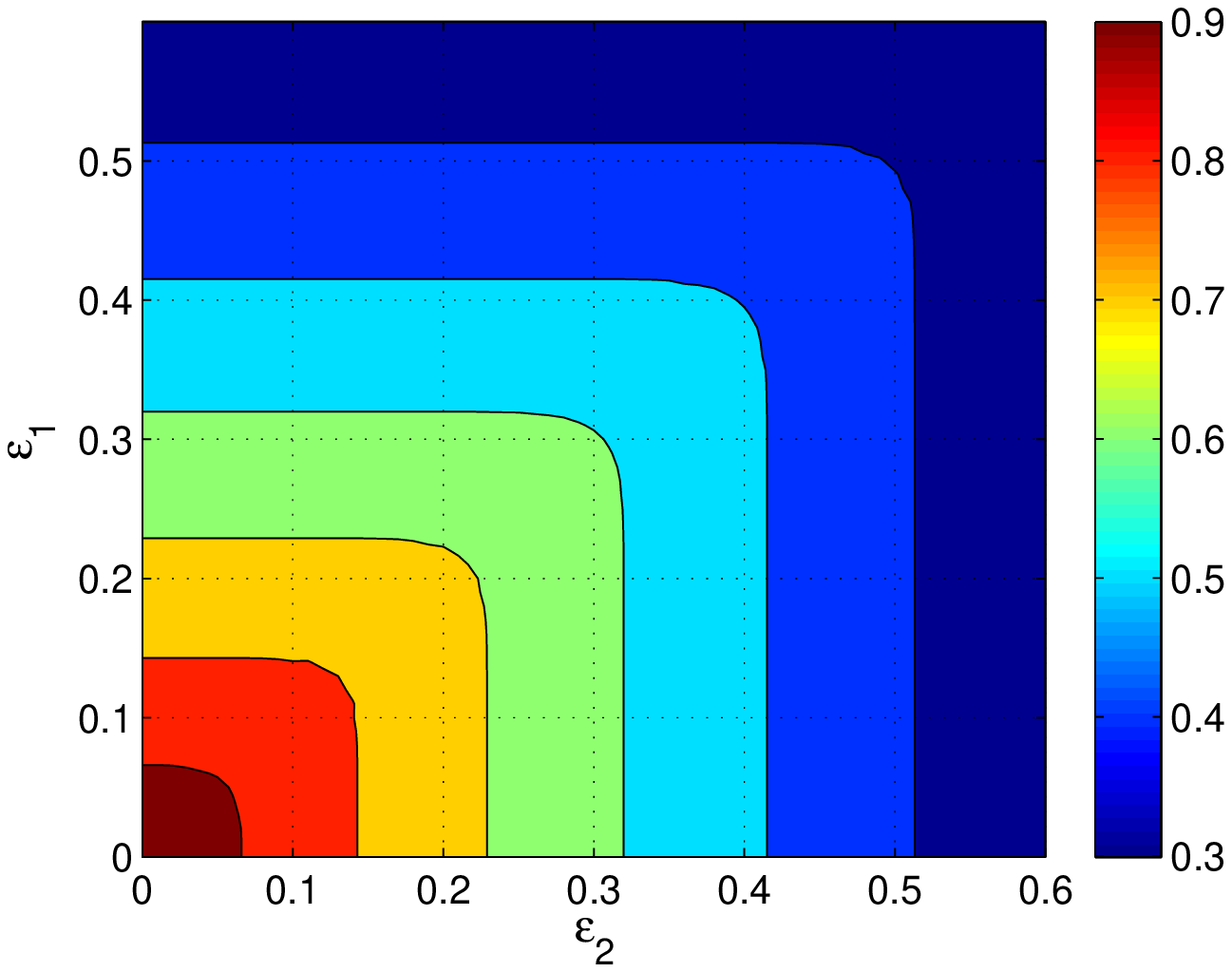}\label{fig:Rnc_5_r_03}}
     \caption{Per-link achievable rate region with respect to link erasure rates for $\eta_{0}=5\%$, $r\in[0.3,1]$.}
     \label{fig:R_5_r_03}
\end{figure*}
Figures \ref{fig:R_5_r_05} and \ref{fig:R_15_r_05} depict achievable rate region of the two coding schemes w.r.t. various values of link erasure rates for different $\eta_0$, where coding rate is limited in range from $0.5$ to $1$. For each value of ($\epsilon_1,\epsilon_2$), we choose a coding rate as large as possible so that $\eta^{e2e}$ and $\eta^{NC}$ satisfy $\eta_0$, respectively with end-to-end coding and coding at the source and re-encoding at the intermediate node. We term achievable region as the region where $\eta_0$ is satisfied. Otherwise, the region in blue represents the cases in which there is not any coding rate in the given range that meets $\eta_0$.

As denoted in Fig. \ref{fig:Re2e_5_r_05}, $R^{e2e}$ reaches $0.9$ since channel condition is good. However, since link erasure rates increase, respective regions such as red, orange, etc show the reduced $R^{e2e}$. This is due to lower coding rate required to meet $\eta_0$ in worse channel conditions. It is interesting to note that each region is shaped by curves. On the other hand, Fig. \ref{fig:Rnc_5_r_05} denotes $R^{NC}$ w.r.t. link erasure rates with $\eta_0=5\%$ where achievable regions are shaped by squares. In general, coding at the source and re-encoding at the intermediate node can widen the achievable region to approximately $90\%$ if compared to end-to-end coding. The detailed proof can be seen in Appendix \ref{app:RateRegion}.
Consider Fig. \ref{fig:R_15_r_05} in which target reliability is now changed to $\eta_0=15\%$, we observe that the shapes of the contours are the same as in Fig. \ref{fig:R_5_r_05}. However, due to the increase of $\eta_0$, the source only needs a smaller value of redundant packets to satisfy the requirement, thus achievable rate regions are extended if compared with those of Fig. \ref{fig:Re2e_5_r_05}.

In addition, in Fig. \ref{fig:R_5_r_03}, regions of $R^{e2e}$ and $R^{NC}$ have been extended by appearance of some new regions located at points with higher $\epsilon_i$ $(i=1,2)$ if compared with those of Figures \ref{fig:R_5_r_05}-\ref{fig:R_15_r_05}. Main reason is that with a wider range of coding rate, regions with very bad channel conditions which $\eta_0$ cannot be satisfied by any $r\in[0.5,1]$ are now covered due to a wider range of feasible coding rate.

In general, we can conclude that 	
	\begin{itemize}
		\item The shapes of $R^{e2e}$ and $R^{NC}$ are curves and squares, respectively, regardless the constraints of $r$ and $\eta_0$. The difference is the extension of achievable rate region when a wider range of coding rate is allowable.
		\item Achievable rate region obtained with NC case is almost twice wider than that of end-to-end coding for the same limitation of coding rate and $\eta_0$.
		\item The shapes of $R^{NC}$ and $R^{e2e}$ are symmetric over the diagonal. Therefore, $R^{NC}$ and $R^{e2e}$ are independent of the  order of the two links. 
	\end{itemize}	

Our numerical results indicate that NC at the source and re-encoding at the intermediate nodes can enhance significantly network performance in terms of achievable rate. 
Therefore, we will apply NC whenever an intermediate node is ready to implement NCF. For the latter, in the specific application of geo-controlled network connectivity, we evaluate the optimized utility, reliability, and connectivity gain beyond the service coverage area of satellite with assumption that NC is deployed both at the source and intermediate nodes along the path.
\subsection{Geo-controlled network connectivity}
In this section, we evaluate the optimized utility, reliability, computational complexity, and connectivity gain for the case of using our design of VGNCF with respect to the number of hops for communication services beyond the coverage area of satellite in low and high complexity constraints. Databases with geo-tagged link statistics and geo-location information are utilized for the optimization functionality towards the energy-efficient use of network resources.
The proposed solution can cope with technical problems such as providing transmission to devices at low power locations (e.g. long distance and/or obstacles) or to extend the transmission beyond the coverage of satellite. We evaluate the potential of our solution through simulation results using Matlab. 

For illustrative purposes, we assume that a great number of network devices are uniformly distributed in a deployment area. %for two different densities XX and YY according to low and high densities, respectively. Maximum radio range (e.g. WiFi signal) of each device is assumed to be $80 m$.
All links undergo the same erasure rate for different cases of $\epsilon=0.1$ and $0.15$, while $k=50$ information packets, $M=100$ bytes, $q=8$, $\rho_{0}=80\%$ for two different computational complexity constraints, $\beta_{0}^{low}=125000$ and $\beta_{0}^{high}=150000$ operations. Coding rate is optimized according to Section \ref{sec:OPT}.
\subsubsection{Network Reliability and Computational Complexity:}
\label{sec:Reliability}
We conduct various numerical results to evaluate how network performance will be improved with VGNCF according to system limitations and different network conditions such as erasure rate, path length, etc.
\begin{figure}
     \centering
     \includegraphics[scale =0.65]{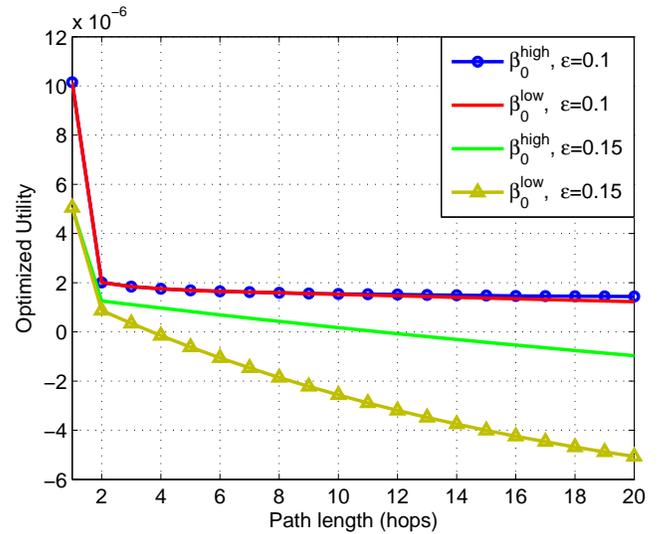}
     \caption{Optimized utility of the source for different path length and erasure rates according to low and high complexity constraints, with $\rho_0=80\%$.}
     \label{fig:UTILITY}
\end{figure}

First, we show the optimized utility of the source for different path length and erasure rates according to low and high complexity constraints. As denoted in Fig. \ref{fig:UTILITY}, the optimized utility is declined w.r.t. the length of path. This is because of that the longer path, the more residual erasure rate accumulated along the path. The source thus needs more redundancy, i.e. additional complexity cost, to cope with the impact of residual erasure rate. In particular, in case of $\beta_{0}^{low}$ with $\epsilon=0.15$, the computational constraint only ensures the design target $\rho_{0}=80\%$ for up to a path of three hops. %Therefore, the optimized utility is minus values beyond that. 
In general, the higher constraint may provide better utility for the source. As for $\beta_{0}^{high}$ with $\epsilon=0.1$, it is clear to see the affect of $\beta_{0}^{high}$. For some twenty hops, the utility is still obtained an optimized value with reliability satisfying the design target. 
\begin{figure}
     \centering
     \includegraphics[scale =0.48]{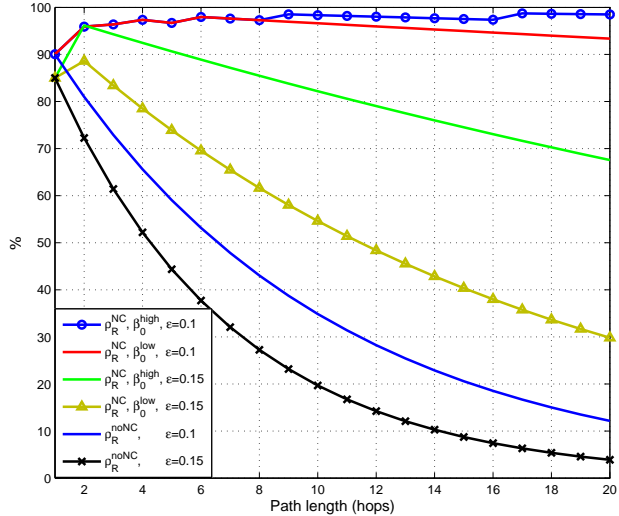}\label{fig:reliability}     
     \caption{Reliability with NC, $\rho_R^{NC}$, and uncoded case, $\rho_R^{noNC}$, according to optimized utility respectively for different path length and erasure rates, with $\rho_0=80\%$.}
     \label{fig:Reliability}
\end{figure}

A glance at Fig. \ref{fig:Reliability} reveals that network reliability with NCF outperforms than that of the case of transmission without NC. Note that in these plots, we only denote the path length up to $20$ hops in which the performance of no-NC case is extremely low and thus $20$ hops are large enough for our comparison. 
In particular, $\beta_{0}^{high}$ can guarantee the target $\rho_{0}=80\%$ for a path length of more than $20$ hops (even some hundred hops) in low erasure rate. For longer paths, the NCF just needs to increase the redundant level of combined packets in order to cope with the residual erasure rate.
In comparison with NC case, network performance without NC (no-NC case) degrades dramatically with the number of hops. 
Especially, consider such high complexity constraints, the target is only satisfied up to approximately $12$ hops in bad conditions e.g. $\epsilon=0.15$. The reason is that the longer the transmission path, the lower the connectivity due to physical limits. The limitations of redundant combined packets then cannot provide the packet successfully-received after decoding as the design target. Even though maximum redundancy is chosen, the utility function is not possible to reach the maximum point. 
%\subsubsection{Computational complexity}
\begin{figure}
     \centering
     \includegraphics[scale =0.65]{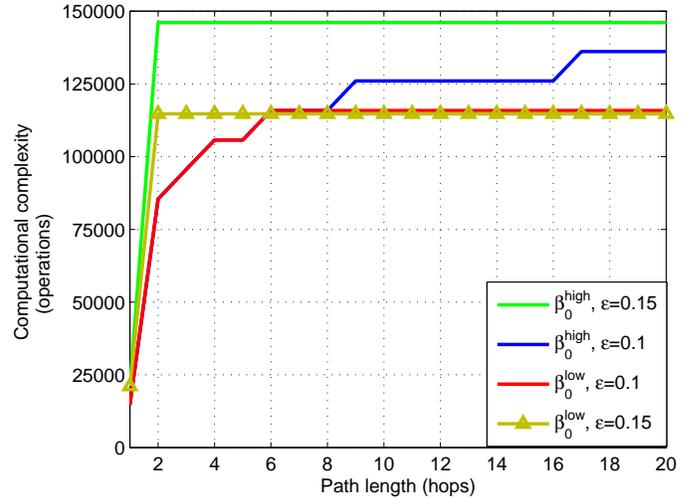}\label{fig:Complexity_Operations}
     \caption{Computational complexity required by the intermediate nodes w.r.t. different erasure rates for low and high complexity constraints, with $\rho_0=80\%$.}
     \label{fig:COMPLEXITY}
\end{figure}

As we have seen, the target reliability can be satisfied for different erasure rates depending on the redundancy of coded packets generated by the source, i.e. some computational cost is needed for activating NCF and thus equivalently requiring more energy consumption at network nodes.
In Fig. \ref{fig:COMPLEXITY}, we evaluate computational complexity at intermediate nodes, $\beta_{R_j}(r)$, which is defined as the total of both encoding and decoding computational complexity.
As showed in the figure, the interesting point is that for $\epsilon=0.15$, computational complexity of NCF for the path of more than $2$ hops is kept constant for both high and low complexity constraints. This is indicated in Fig. \ref{fig:Reliability} that for such high erasure rate, the source has to choose the highest values of redundant combinations to maximize its utility. Therefore, computational complexity is the highest value corresponding to those constraints. 
On the other hand, for the lower erasure rate, $\epsilon=0.1$, $\beta_{0}^{high}=150000$ is large enough to obtain maximal value of utility function. For the increasing of path length, the source gradually adds more redundant combinations to cope with the accumulating residual erasure rate as we can observe in Fig. \ref{fig:Reliability}. The reliability is almost stable while the complexity is gradually increasing.
\subsubsection{Connectivity gain:}
\vspace{-6pt}
Assuming some reliability design target $\rho_{0}$, in the uncoded case, many nodes would not achieve the target $\rho_{0}$ while NC case could. Let $h^{NC}(\rho_{0})$ and $h^{noNC}(\rho_{0})$ denote the hop at which NC and the uncoded case can provide connectivity with the reliability satisfying $\rho_{0}$, respectively. For simplicity, the connectivity gain for a target $\rho_{0}$ is defined as $\gamma\left(\rho_{0}\right)=h^{NC}(\rho_{0})/h^{noNC}(\rho_{0})$.
\begin{figure*}
     \centering
     \subfloat[][$\rho_0=80\%$]			{\includegraphics[scale =0.50]{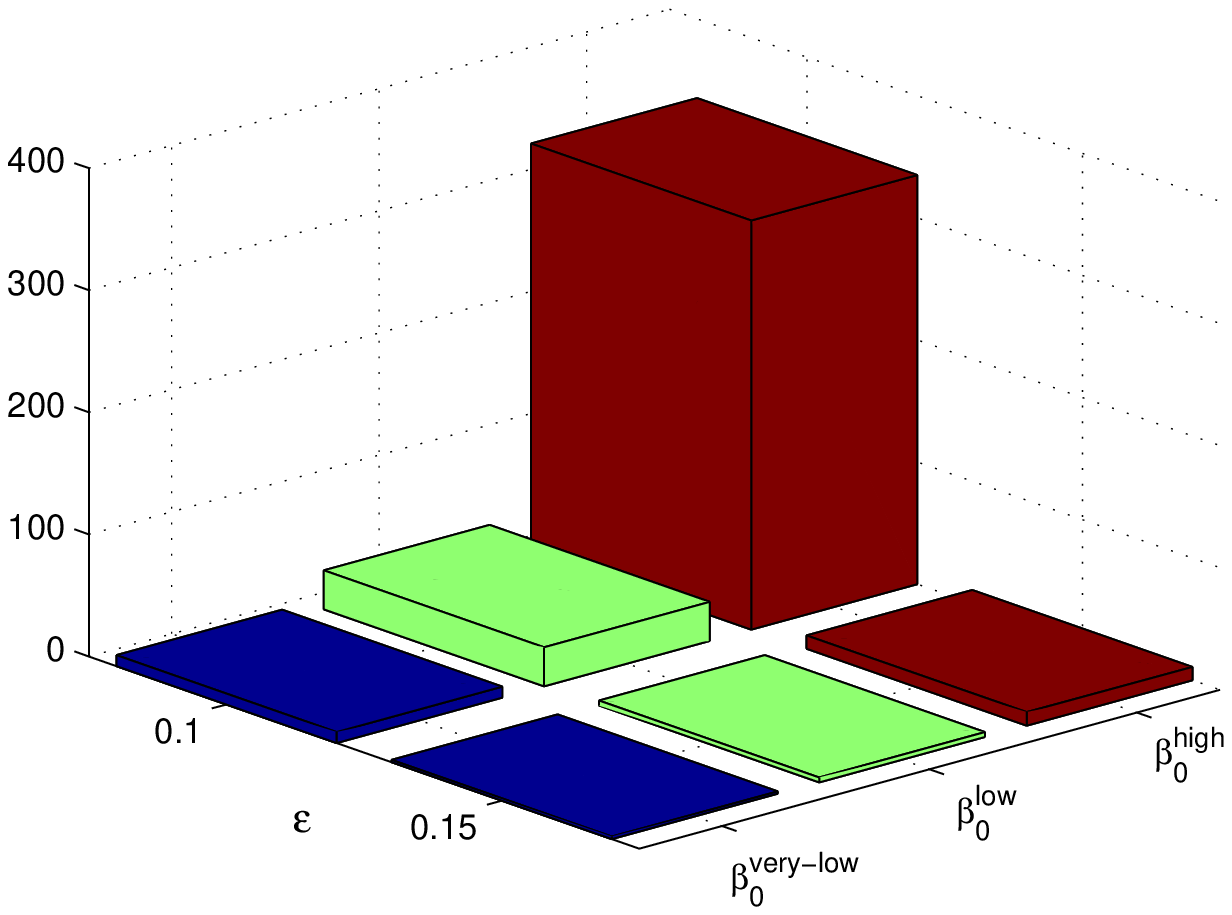}\label{fig:gain_PDR80}}
     \subfloat[][$\rho_0=85\%$]			{\includegraphics[scale =0.50]{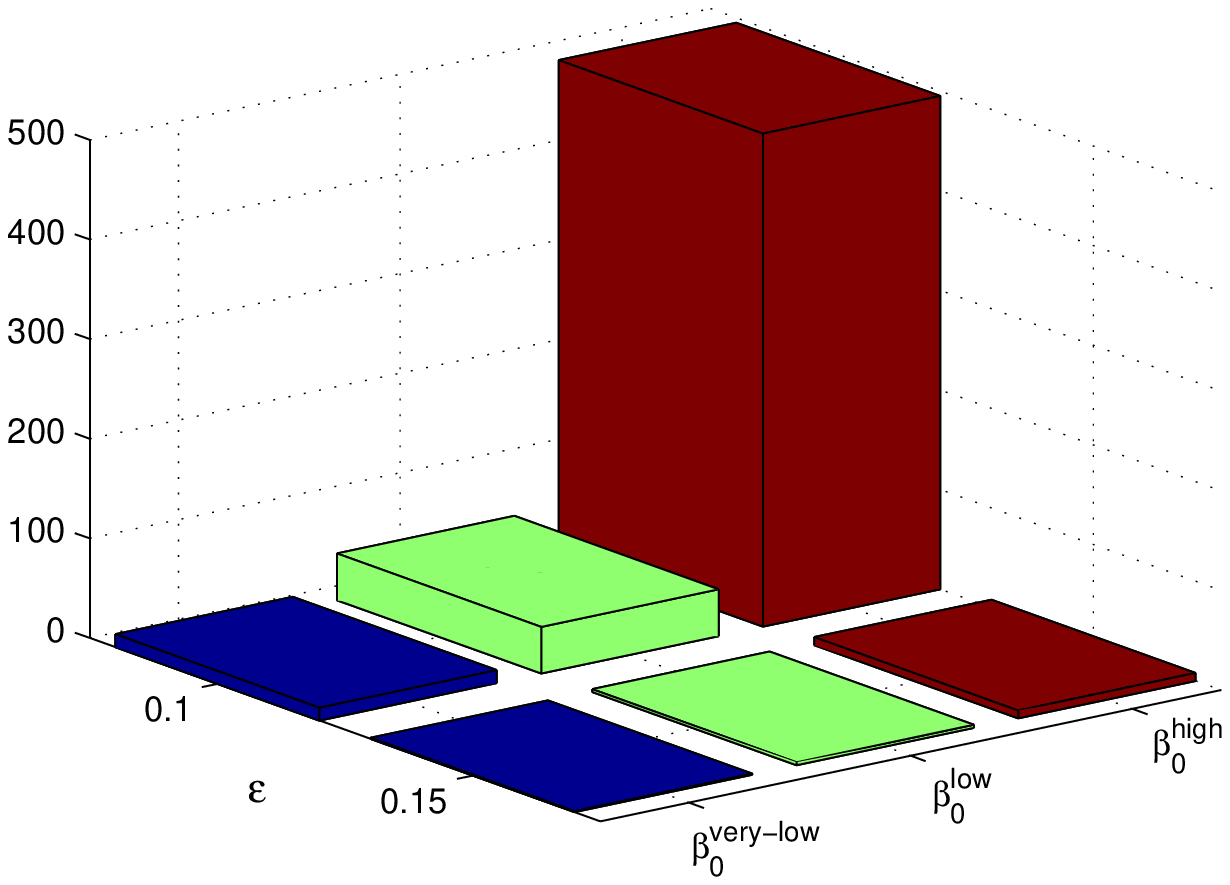}\label{fig:gain_PDR85}}
     \caption{Connectivity gain when using optimized NC in times for different $\rho_0$ and $\beta_0$.}
     \label{fig:CONNECTIVITY}
\end{figure*}

Fig. \ref{fig:CONNECTIVITY} depicts connectivity gain when using NCF in times for different $\rho_0$ with $\beta_0^{very-low}=100000$, $\beta_0^{low}=125000$, and $\beta_0^{high}=150000$. The larger the computational constraint, the higher the connectivity gain beyond the cell coverage of the satellite. Particularly, since $\epsilon=0.1$, NC can obtain up to $335$ times and $495$ times gain in connectivity if compared to the uncoded case with high constraint for $80\%$ and $85\%$ target reliability, respectively. Otherwise, in low constraint, the connectivity gain obtained with NC is $32$ times and $47$ times for $80\%$ and $85\%$ reliability, respectively. The reason is that the performance of the uncoded case is significantly impacted by per-link erasure rate and length of transmission path. NC case, meanwhile, can adapt its coding rate within the constraints to obtain the target reliability while optimizing source's utility. However, implementing NCF is then strongly affected by computational complexity which is equivalently mapped into the price to pay in terms of energy consumption.

%------------------------------------------------------------------------------------%
\section{Conclusions}
\label{sec:Conc}
\vspace{-2pt}
In this paper, we have proposed the integration of NC and NFV architectural design. Particularly, NC functionalities have been identified as a toolbox so that NC can be designed as a virtual network function thus providing flow engineering functionalities to the network.
In addition, SNC and novel subspace coding scheme have been introduced as part of coding design domain to make it possible efficient-use of network codes for coherent and non-coherent design approaches, respectively.

We have conducted a complete design to illustrate the use and relevance of our proposed VGNCF design where geographical information is the key enabler to support VGNCF.
Specifically, we make it clear the interaction between our proposed VGNCF and NFV-MANO, the brain of NFV architecture, through a general procedure for the instantiation, performance monitoring, and termination of VGNCF.
Especially, optimization functionality ensures an optimized and energy-efficient operation of VGNCF. Several numerical results show the improvement of overall throughput (achievable rate) with NC at the source and re-encoding at the intermediate nodes if compared with transmission with NC at the source only. Furthermore, the service-driven optimized VGNCF can provide connectivity gain up to $32$ times and $47$ times in comparison with the uncoded case for $80\%$ and $85\%$ reliability with complexity constraint of $125000$ operations. 

Our proposed framework can naturally be tailored for different designs and accommodate additional functionalities. Moreover, the proposed architectural design framework and interaction between VGNCF and NFV-MANO functional blocks can also be generalized as a design framework for several different VNFs.
This is only the first step towards a full integration of VGNCF architecture design into ETSI NFV architecture. In the future, we plan to investigate implementation and deployment of VGNCF design for next generation networks, e.g. $5G$ networks.

%------------------------------------------------------------------------------------%
\begin{center}
\textbf{APPENDICES}
\end{center}
\appendix
%------------------------------------------------------------------------------------%
%\section{Systematic Network Coding \cite{Saxena.2016}} 
%\label{app:SNC}

\section{Proof of results in Section \ref{sec:RateRegion}}
\label{app:RateRegion}
\subsection{End-to-end coding}
$\eta^{e2e}$ depends on end-to-end packet erasure rate given as $\eta^{e2e}=1-(1-\epsilon_1)(1-\epsilon_2)$. 
$\eta^{e2e}$  and $R^{e2e}$ thus directly depend on the product $(1-\epsilon_1)(1-\epsilon_2)$. 
Therefore, we will consider whether there exists a set of $(\epsilon_1,\epsilon_2)$ with the same value of the product $(1-\epsilon_1)(1-\epsilon_2)$ leading the same $R^{e2e}$. 

We evaluate the function $f(\epsilon_1,\epsilon_2) =(1-\epsilon_1)(1-\epsilon_2)$ to address the existence of such set of $(\epsilon_1,\epsilon_2)$ and how they affect the shapes of $R^{e2e}$ as depicted in Figures \ref{fig:Re2e_5_r_05}-\ref{fig:Re2e_5_r_03}.

We assume that there is a set of $(\epsilon_1,\epsilon_2)$ so that $f(\epsilon_1,\epsilon_2)$ is a given constant, e.g. $(1-\epsilon_1 )(1-\epsilon_2 )=A$, for some $(\epsilon_1,\epsilon_2)$ in a limited range e.g. $\epsilon_i\in[0,0.6]$  $(i=1,2)$. For example, for $(\epsilon_1=a,\epsilon_2=0)$, we find $f(\epsilon_1=a,\epsilon_2=0) = A$. Then, our concern is that whether there is a set of $(\epsilon_1,\epsilon_2)$ so that value of the product $f(\epsilon_1,\epsilon_2)$ is still equal to $A$. If $(1-\epsilon_1)(1-\epsilon_2)=A$ and an arbitrary $\epsilon_1$, it necessarily exists $\epsilon_2$ written by 
\begin{eqnarray}
\epsilon_2=f(\epsilon_1 )=\frac{\epsilon_1+A-1}{\epsilon_1-1}.
\label{eq:E8}
\end{eqnarray}

Note that the point $(\epsilon_1=a,\epsilon_2=0)$ satisfies Eq. \ref{eq:E8} and there are infinitely such points. Analyzing Equation \ref{eq:E8}, a form of reciprocal graphs, we obtain numerical results as denoted in Fig. \ref{fig:why_E2E_curve}.
\begin{figure}[ht]
\centering
\includegraphics[scale =0.6] {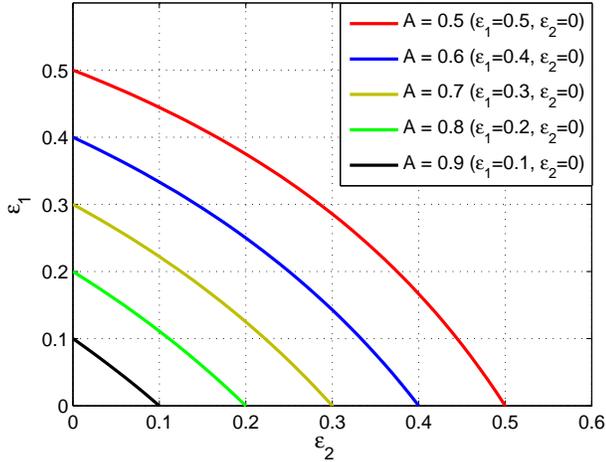}
\caption{Illustration of some values of $(\epsilon_1,\epsilon_2)$ that brings the same $R^{e2e}$.}
\label{fig:why_E2E_curve}
\end{figure}

Assume that $\epsilon_1\in[0,0.6]$, $\epsilon_2\in[0,0.6]$, we find some values of $A$ respectively to some points of $(\epsilon_1=a,\epsilon_2=0)$, e.g. $(\epsilon_1=0.5,\epsilon_2=0)$, $(\epsilon_1=0.4,\epsilon_2=0)$, $(\epsilon_1=0.3,\epsilon_2=0)$, $(\epsilon_1=0.2,\epsilon_2=0)$, $(\epsilon_1=0.1,\epsilon_2=0)$.
As depicted in Fig. \ref{fig:why_E2E_curve}, it is interesting to see that 

(1) There exists a set of $(\epsilon_1,\epsilon_2)$ for a specific constant $A$ respectively to a point $(\epsilon_1=a,\epsilon_2=0)$. 

(2) The set of $(\epsilon_1,\epsilon_2)$ with respect to $A$ draws a symmetric curve. In other words, this proves the existence of a set of $(\epsilon_1,\epsilon_2)$ producing a same value of $(\epsilon_1,\epsilon_2)$ and therefore, $R^{e2e}$ obtains a same value for the set of $(\epsilon_1,\epsilon_2)$.

(3) The above results also indicate the symmetry of the curves and the independence of the order of the links. Especially, it is clear to observe the shapes of regions at the frontier and at the boundary of achievable rate region as shown in Figures \ref{fig:Re2e_5_r_05}-\ref{fig:Re2e_5_r_03}. Furthermore, as for small values of $(\epsilon_1,\epsilon_2 )$, the curves look as lines within our limitation of $(\epsilon_1,\epsilon_2)$.

\subsection{NC at source and re-encoding at the intermediate node}
We assess the shapes of achievable rate regions based on characteristics of $\eta^{NC}=1-(1-\eta_1)(1-\eta_2)$ by means of numerical results. We rely on the behavior of $\eta_i$  $(i=1,2)$ to address the behavior of $\eta^{NC}$ and $R^{NC}$. 

Due to the independent of the order of the two links, the behavior of $\eta_1$ and $\eta_2$ is similar. We thus keep $\epsilon_1$  (link $1$) fixed while $\epsilon_2$ is chosen in a range. For the sake of simplicity, we first consider an arbitrary point $(\epsilon_1=a,\epsilon_2=0)$ (e.g. $a\in[0,0.6]$) located inside the achievable region of $R^{NC}$. Let $r_0$ be coding rate at that point so that $\eta^{NC}(\epsilon_1=a,\epsilon_2=0)\leq\eta_0$.
Because of $\epsilon_2=0$ for any $r_0$, hence $\eta_2=0$ and $\eta^{NC}=\eta_1$ which reduce the complexity of identifying $r_0$.
Now our main target is to show whether or not  $\eta^{NC}(\epsilon_1=a,\epsilon_2<a) \approx \eta^{NC} (\epsilon_1=a,\epsilon_2=0)$ at $r_0$ identified at previous step. It means that for a given $\epsilon_1=a$ and the respective $r_0$, $\eta^{NC}$ is now approximately stable over a range of $\epsilon_2\in[0,a)$. Equivalently, whether $R^{NC}$ is approximately following the shape of square form or not as depicted in Figures \ref{fig:Rnc_5_r_05}-\ref{fig:Rnc_5_r_03}.
By means of simulation, we show how $\eta^{NC}(\epsilon_1=a,\epsilon_2<a)$ behaves for a range of $\epsilon_2$ with the same $r_0$. Because $\eta_1(\epsilon_1=a)$ is a constant for $r_0$, $\eta^{NC}(\epsilon_1=a,\epsilon_2<a)$ is thus only affected by $\eta_2(\epsilon_2<a)$. For the increase of $\epsilon_2$, $r_0$ is still kept if $\eta^{NC}(\epsilon_1=a,\epsilon_2<a)\leq\eta_0$.

In Figures \ref{fig:why_NC_square_e01} and \ref{fig:why_NC_square_e03}, we indicate the behavior of $\eta_2$, $\eta^{NC}$, and $R^{NC}$ with $\eta_0=5\%$, $r\in[0.3,1]$, and $\epsilon_1=0.1, 0.3$, respectively. We identify optimal coding rate $r_0=0.88,0.62$, respectively for each case of $(\epsilon_1=0.1,\epsilon_2=0)$, $(\epsilon_1=0.3,\epsilon_2=0)$ so that the target $\eta_{NC}\leq\eta_0$ is satisfied. 
In particular, Fig. \ref{fig:why_NC_square_e01} shows the behavior of $\eta_2(.)$ according to $\epsilon_1=0.1$, $\epsilon_2\in[0,0.12]$ and $r_0=0.88$. It is noted that $\eta^{NC}(\epsilon_1=0.1,\epsilon_2<0.08)\leq\eta_0$. This is because of that $\eta_2 (\epsilon_2<0.08) \ll \eta_2 (\epsilon_2=0.1)$, e.g. $\eta_2(\epsilon_2=0.07)=0.004 \ll \eta_2 (\epsilon_2=0.1)=0.038$. Therefore,  $1-\eta_2 (\epsilon_2\leq0.07) \approx 1$. Hence, $\eta^{NC}(\epsilon_1=0.1,\epsilon_2\leq0.07) \approx \eta^{NC}(\epsilon_1=0.1,\epsilon_2=0)$. This estimation explains why $R^{NC}$ represented by the green line is almost stable during $\epsilon_2\in[0,0.07]$. 
In other words, the residual erasure rate of NC, $\eta^{NC}(\epsilon_1,\epsilon_2<\epsilon_1)$ is almost decided by $\epsilon_1$, i.e. the impact of $\epsilon_1$ is dominant that of $\epsilon_2$.
This explains the shape of $R^{NC}$, which follows the square form.
However, for $\epsilon_2\in(0.07,0.1]$ approximately $30\%$ of $\epsilon_i=0.1$, $\eta_2$ is comparable with $\eta_1$, the coding rate thus should be reduced to keep $\eta^{NC}$ low as the requirement. $R^{NC}$ is thus changed significantly. Hence, we can infer that $R^{NC}$ follows the square form but imperfect form for some percent of $\epsilon_1$ and $\epsilon_2$. For instance, in this case, this range is approximately $30\%$ of $\epsilon_i=0.1$ $(i=1,2)$.
\begin{figure*}[ht]
\centering
\includegraphics[scale =0.32] {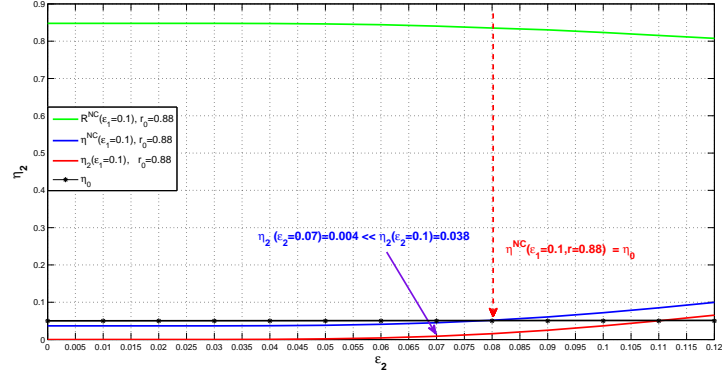}
\caption{Illustration of $\eta_2$, $\eta^{NC}$, and $R^{NC}$ with $\epsilon_1=0.1$ and $\eta_0=5\%$, $r\in[0.3,1]$, $\eta^{NC}(\epsilon_1=0.1,\epsilon_2=0)\leq\eta_0$ at $r_0=0.88$.}
\label{fig:why_NC_square_e01}
\end{figure*}
\begin{figure*}[ht]
\centering
\includegraphics[scale =0.32] {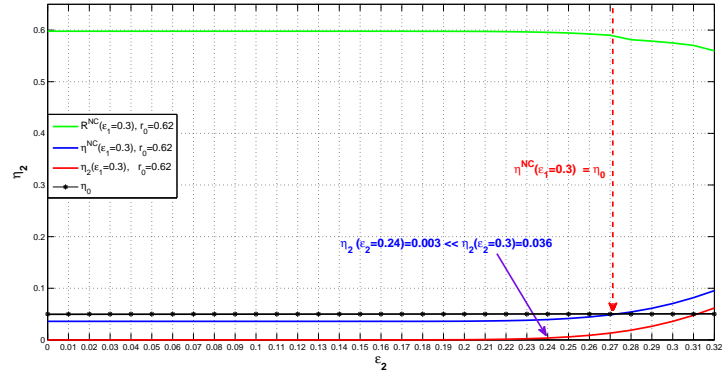}
\caption{Illustration of $\eta_2$, $\eta^{NC}$, and $R^{NC}$ with $\epsilon_1=0.3$ with $\eta_0=5\%$, $r\in[0.3,1]$, $\eta^{NC}(\epsilon_1=0.3,\epsilon_2=0)\leq\eta_0$ at $r_0=0.62$.}
\label{fig:why_NC_square_e03}
\end{figure*}

For another example, Fig. \ref{fig:why_NC_square_e03} depicts the behavior of $\eta_2(.)$ according to $\epsilon_1=0.2$, $\epsilon_2 \in [0,0.22]$ and $r_0=0.74$. It is clear to see that $\eta^{NC}(\epsilon_1=0.2,\epsilon_2 \leq 0.18) \leq \eta_0$ and $\eta_2(\epsilon_2<0.16) \ll \eta_2(\epsilon_2=0.2)$, e.g. $\eta_2(\epsilon_2=0.15)=0.003 \ll \eta_2(\epsilon_2=0.2)=0.033$. Therefore, $1-\eta_2(\epsilon_2 \leq 0.15) \approx 1$. Hence, $\eta^{NC}(\epsilon_1=0.2,\epsilon_2 \leq 0.15) \approx \eta^{NC}(\epsilon_1=0.2,\epsilon_2=0)$. Herein, $R^{NC}$ is stable for a range of $\epsilon_2\in[0,0.15]$, which is illustrated by the green line. Therefore, the shape of $R^{NC}$ follows the square form for a range of $(\epsilon_1\in[0,0.15], \epsilon_2\in[0,0.15])$. The remaining $25\%$ of $\epsilon_i=0.2$ $(i=1,2)$ does not follow the square form.

In general, we can conclude that for a limitation of coding rate, the larger the $\epsilon_i$  $(i=1,2)$, the better the shapes of $R^{NC}$ following the rectangle form (a more general form assuming that $\epsilon_1\neq\epsilon_2$) that we can see in the achievable rate region, especially at the boundary of the region. 
This reflects exactly the shapes of $R^{NC}$ as depicted in Figures \ref{fig:Rnc_5_r_05}-\ref{fig:Rnc_5_r_03}, which are in terms of squares although it is not really perfect at the corner. Furthermore, due to the independence of the order of the links, the shape of $R^{NC}$ is symmetric over the diagonal of the two axes.

%------------------------------------------------------------------------------------%
% Use IEEE bibliography type
\bibliographystyle{IEEETran}

% Bibliography file
\bibliography{satcom}

\end{document}